\documentclass[english]{IEEEtran}
\usepackage[T1]{fontenc}
\usepackage[latin9]{inputenc}
\usepackage[active]{srcltx}
\usepackage{float}
\usepackage{amsmath}
\usepackage{amsthm}
\usepackage{amssymb}
\usepackage{graphicx}
\usepackage{esint}

\makeatletter

\providecommand{\tabularnewline}{\\}
\floatstyle{ruled}
\newfloat{algorithm}{tbp}{loa}
\providecommand{\algorithmname}{Algorithm}
\floatname{algorithm}{\protect\algorithmname}

  \theoremstyle{definition}
  \newtheorem{defn}{\protect\definitionname}
  \theoremstyle{remark}
  \newtheorem{claim}{\protect\claimname}
  \theoremstyle{plain}
  \newtheorem{thm}{\protect\theoremname}
  \theoremstyle{plain}
  \newtheorem{conjecture}{\protect\conjecturename}
 \ifx\proof\undefined\
   \newenvironment{proof}[1][\proofname]{\par
     \normalfont\topsep6\p@\@plus6\p@\relax
     \trivlist
     \itemindent\parindent
     \item[\hskip\labelsep
           \scshape
       #1]\ignorespaces
   }{%
     \endtrivlist\@endpefalse
   }
   \providecommand{\proofname}{Proof}
 \fi
  \theoremstyle{plain}
  \newtheorem{lem}{\protect\lemmaname}

\usepackage{babel}

\makeatother

\usepackage{babel}
\providecommand{\claimname}{Claim}
\providecommand{\conjecturename}{Conjecture}
\providecommand{\definitionname}{Definition}
\providecommand{\lemmaname}{Lemma}
\providecommand{\theoremname}{Theorem}

\begin{document}

\title{Physical-Layer Cryptography \hspace{50mm}Through Massive MIMO}

\author{Thomas R. Dean,\emph{ Member, IEEE,} and Andrea J. Goldsmith\emph{,
Fellow, IEEE} \thanks{The authors are with the Department of Electrical Engineering, Stanford
University, 350 Serra Mall, Stanford, CA 94305 (e-mail: trdean@stanford.edu;
andrea@ee.stanford.edu).}\thanks{This work was presented in part at the 2013 IEEE Information Theory
Workshop in Seville, Spain. T. Dean is supported by the Fannie and
John Hertz Foundation. This work was supported in part by the NSF
Center for Science of Information under Grant CCF-0939370.}}
\maketitle
\begin{abstract}
We propose the new technique of physical-layer cryptography based
on using a massive MIMO channel as a key between the sender and desired
receiver, which need not be secret. The goal is for low-complexity
encoding and decoding by the desired transmitter-receiver pair, whereas
decoding by an eavesdropper is hard in terms of prohibitive complexity.
The decoding complexity is analyzed by mapping the massive MIMO system
to a lattice. We show that the eavesdropper's decoder for the MIMO
system with M-PAM modulation is equivalent to solving standard lattice
problems that are conjectured to be of exponential complexity for
both classical and quantum computers. Hence, under the widely-held
conjecture that standard lattice problems are hard to solve, the proposed
encryption scheme has a more robust notion of security than that of
the most common encryption methods used today such as RSA and Diffie-Hellman.
Additionally, we show that this scheme could be used to securely communicate
without a pre-shared secret and little computational overhead. Thus,
by exploiting the physical layer properties of the radio channel,
the massive MIMO system provides for low-complexity encryption commensurate
with the most sophisticated forms of application-layer encryption
that are currently known.

\smallskip{}

\end{abstract}

\section{Introduction}

The decoding of massive MIMO systems forms a complex computational
problem. In this paper, we exploit this complexity to form a notion
of physical-layer cryptography. The premise of physical-layer cryptography
is to allow the transmission of confidential messages over a wireless
channel in the presence of an eavesdropper. We present a model where
a given transmitter-receiver pair is able to efficiently encode and
decode messages, but an eavesdropper who has a physically different
channel must perform an exponential number of operations in order
to decode. This allows for confidential messages to be exchanged without
a shared key or key agreement scheme. Rather, the encryption exploits
physical properties of the massive MIMO channel. 

Our MIMO wiretap channel model for communication is shown in Figure
1. Here, a parallel channel decomposition allows for two users, Alice
and Bob, to communicate with an overhead of only performing linear
precoding and receiver shaping of their MIMO channel, assumed known
to both of them. To an eavesdropper, Eve, who has a different channel,
this decomposition does not aid in the ability to decode the channel
with linear complexity. In particular, we prove that it is exponentially
hard\footnote{More precisely, we prove that it is at least as hard as solving standard
lattice problems in the worst case, which are conjectured to be exponentially
hard.} for the eavesdropper to decode Alice's transmitted vector in our
system model even if it knows the channel between Alice and Bob. We
refer to the channel encryption key as a \emph{Channel State Information-}
or \emph{CSI-key}. The model requires both the transmitter and receiver
to have perfect knowledge of the channel, but this knowledge does
not need to be kept secret. For decoding by Eve to be hard, our model
requires a maximum on the SNR that Eve maintains and that Alice and
Bob use a large constellation size, where the required constellation
size is related to the number of transmit antennas. 

\begin{figure}
\begin{centering}
\includegraphics[width=1\columnwidth]{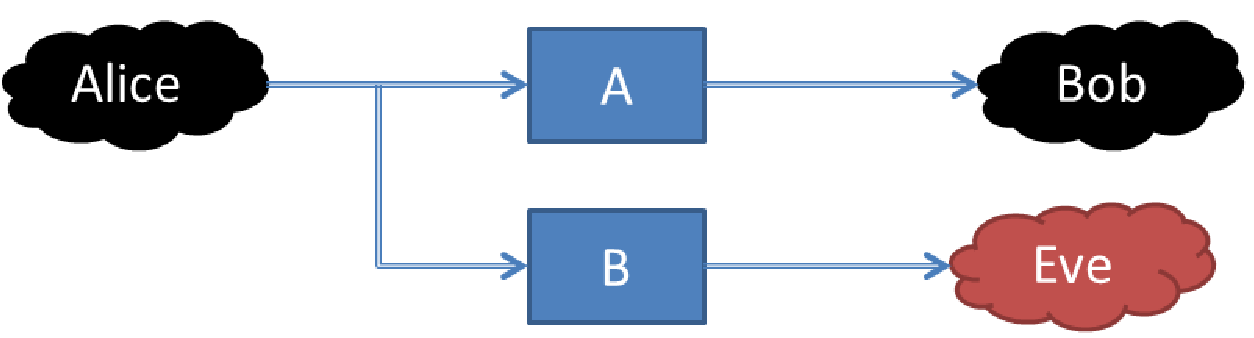} 
\par\end{centering}
\caption{A MIMO wiretap channel model, defined by a channel gain matrix $\mathbf{A}=\mathbf{U\Sigma V^{H}}$,
where $\mathbf{A}$ is known to both Alice and Bob. This allows Bob
to efficiently decode Alice's message. If Eve is not physically co-located,
then knowledge of $\mathbf{A}$, which we call the channel state information
key, does not aid her in decoding the message with low complexity.
Through the use of reductions we show that the complexity of Eve decoding
Alice's message to Bob is at least as hard as standard lattice problems.
Hence, this complexity is conjectured to be exponentially hard in
the number of transmitter antennas Alice uses. In particular, no existing
algorithms, including those of a quantum computer, have been shown
to solve such problems in sub-exponential time}
\end{figure}

To characterize the hardness of decoding MIMO systems, we use the
method of reductions found in the computational study of cryptography.
A brief description of this method will be provided below, while a
more detailed overview of this approach to cryptography can be found
in \cite{KL07}. In the method of reductions, we suppose that we have
access to an oracle, i.e., a black box which, given a channel and
a received vector, returns the proper transmitted vector in a single
operation. If the existence of such an oracle would imply efficient
solutions (i.e., ones whose runtime is bounded to within a polynomial
factor of the length of the input) to problems which are known (or
conjectured) to be hard (i.e., ones whose runtime is bounded to within
an exponential factor of the length of the input), then it follows
that such an oracle cannot (or is conjectured to not) exist. 

We show in this work that under proper conditions, the complexity
of decoding MIMO systems by an eavesdropper can be related to solving
standard lattice problems. The connection between MIMO and lattices
is not new: for example, see Damen et al. \cite{DGC03}, where the
maximum-likelihood decoder is related to solving the Closest Vector
Problem. Problems on lattices have been widely studied in cryptography
and other fields and many standard lattice problems are conjectured
to be hard ({[}3{]}\textendash {[}13{]}). Lattice problems are also
conjectured to be hard even when solved using quantum computers when
they exist (\cite{MR09}, \cite{R05}). Creating efficient cryptosystems
that achieve security given the presence of quantum computers is currently
an active area of research, since most cryptosystems today, such as
RSA (a public-key algorithm named after its inventors, Rivest, Shimir,
and Adleman, \cite{RSA}) or the Diffie-Hellman key exchange (see
\cite{DH}), could efficiently be broken by a quantum computer ({[}16{]}\textendash {[}18{]}).
Our physical-layer cryptosystem provides quantum-resistant cryptography
by exploiting the hardness associated with the eavesdropper's decoding
in a massive MIMO system.

The idea of exploiting properties of the physical layer to achieve
secrecy is not new and dates back to Shannon's notion of information
theoretic secrecy \cite{S49} and Wyner's wiretap channel \cite{W75}
(for a survey on the subject see \cite{MFHS10}). In information theoretic
secrecy, the goal is to communicate in a manner such that legitimate
users may communicate at a positive rate, while the mutual information
between the eavesdropper and the sender is negligibly small. Note
that since in our model Bob and Eve have statistically-identical channels,
information theoretic secrecy is not possible without a key rate,
coding, or using properties of Alice's and/or Bob's channel in the
transmission strategy. Along these lines, Shimizu et al. \cite{SISP10},
have suggested using properties of the radio propagation channel to
achieve information theoretic secrecy. These notions all differ from
ours as we consider encryption based on computational complexity at
the eavesdropper's decoder rather than through equivocation at the
eavesdropper related to entropy and mutual information. We believe
this work makes significant progress in addressing some of the challenges
that have been identified in applying physical layer security to existing
and future systems \cite{Trappe}.

The remainder of this paper is organized as follows. Section II outlines
our system model and the underlying assumptions upon which the security
of our system is based. In Section III we discuss lattices, lattice
problems, and lattice-based cryptography in order to provide the background
for our main result, which is stated in Section IV and proved in Appendix
A. In Section V we discuss additional notions of security for our
model, including how to achieve security under adversarial models
commonly considered by cryptographers. 

\section{Problem Formulation}

\subsection{The Wiretap Model}

Consider an $n\times m$ real-valued MIMO system consisting of $n$
transmit antennas and $m$ receive antennas: 
\begin{equation}
\mathbf{y=Ax+e},
\end{equation}
where $\mathbf{x}\in\mathbb{R}^{n}$, and $\mathbf{A}\in\mathbb{R}^{n\times m}$
is the channel gain matrix. Each entry of the channel gain matrix
is drawn i.i.d. from the Gaussian distribution with zero mean and
standard deviation $k/\mbox{\ensuremath{\sqrt{2\pi}}}$. This distribution
is henceforth written $\Psi_{k}$. The vector $\mathbf{e}\in\mathbb{R}^{m}$
is the channel noise with each entry i.i.d. $\Psi_{\alpha}$. It is
assumed that $\mathbf{A}$ is known to both the transmitter and receiver.
If we constrain the vector $\mathbf{x}$ to use a discrete, periodic
constellation, then the set of received points becomes analogous to
points on a lattice, perturbed by a Gaussian random variable. 

We assume that there are an arbitrary number of receive antennas,
restricted to an amount within a polynomial factor of the number of
transmit antennas, $n$. By making this assumption, we are considering
the advantage an eavesdropper would gain by having an arbitrarily
large number of receive antennas, but we are assuming that building
a receiver with an exponential number of antennas relative to those
at the legitimate user's transmitter would be prohibitively expensive.
Assuming that certain requirements on SNR and constellation size are
met, as described below, the security of the system can be quantified
solely by the number of transmit antennas. This number plays the role
of the \emph{security parameter} commonly used by cryptographers in
designing cryptosystems. Essentially, we are saying that decoding
the system requires a number of operations that is exponential in
$n$ and we parameterize the remaining variables in our system based
on $n$. 

We consider real systems with the transmitted signal constellation,
$\mathcal{X}$, defined as the set of integers $[0,M)$. Lattices
can easily be scaled and shifted, so we use this constellation without
loss of generality over all possible M-PAM constellations. Our analysis
assumes channels with a zero-mean gain, but this is for simplicity
of exposition as the proof given in Appendix A can easily be generalized
to systems with non-zero mean. We consider uncoded systems, but our
results can be extended to coded systems as we discuss in Section
V. 

Let the vector $\mathbf{a}_{i}\in\mathbb{R}^{n}$ denote the gains
between the transmitter and the $i$th receive antenna, let $e_{i}$
denote the noise sample at this antenna, and let $\mathbf{x}$ represent
the transmitted vector which is drawn from $\mathcal{X}^{n}$. The
$i$th receive antenna gets a noisy, random inner-product of the form

\begin{equation}
y_{i}=\left\langle \mathbf{a}_{i},\mathbf{x}\right\rangle +e_{i}.
\end{equation}

Our work can easily be extended to complex channels, where the real
and imaginary portions of the channel gain matrix and noise are drawn
independently, and M-QAM symbols are sent, using an equivalent constellation
as in the real case.

Our model considers static channels. However, our model can easily
be extended to time-varying channels: consider a single message is
transmitted $t$ times, no more than once per channel coherence period.
Each transmission is equivalent to Eve receiving the original message
with $t$ times as many receive antennas which does nothing to aid
in the computational complexity of her decoding.

Let user $\mathcal{A}$ have $n$ transmit antennas which are used
to send a message to user $\mathcal{B}$ who has a number of receive
antennas that is within a polynomial factor of $n$. Let the channel
between $\mathcal{A}$ and $\mathcal{B}$ be represented by $\mathbf{A}$,
where each entry is i.i.d. Gaussian $\Psi_{k}$. The noise at each
receive antenna is drawn from $\Psi_{M\alpha}$. The results in this
paper show that MIMO decoding can be related to solving standard lattice
problems when a certain minimum noise level and constellation size
are met. If the noise power is below the required level, efficient
decoding methods such as the zero-forcing decoder could be applied
to our system. In other words, if these conditions are not met, then
our results provide no insight on the complexity of decoding, and
hence on the security of the MIMO wiretap channel. Specifically, for
some arbitrary $m>0$, we require the following constraints on the
transmission from user $\mathcal{A}$ to user $\mathcal{B}$:

\begin{equation}
\textrm{Minimum Noise: }m\alpha/k^{2}>\sqrt{n}
\end{equation}

\begin{equation}
\textrm{Constellation Size: }M>m\,2^{n\log\log n/\log n}
\end{equation}
where the parameter $m$ may be chosen by a user or system designer
in order to trade off the SNR requirement for the size of the constellation.

Now consider an eavesdropper, $\mathcal{EVE}$, which has $\mbox{poly}(n)$
receive antennas, and receives message $\mathbf{x}$ with channel
represented by $\mathbf{B}$, where each entry is again i.i.d. $\Psi_{k}$.
Let the channel have noise be drawn from $\Psi_{M\beta}$, where $\beta\geq\alpha$.
In other words, the eavesdropper must meet at least the minimum noise
requirement stated above. Discussion on how to ensure this requirement
is met is also provided in Section V.

In order to send message $\mathbf{x}$ to user $\mathcal{B}$, user
$\mathcal{A}$ performs a linear precoding as described in \cite{G04}.
Let the singular value decomposition of $\mathbf{A}$ be given as
\textbf{$\mathbf{A}=\mathbf{U}\mathbf{\Sigma}\mathbf{V}^{H}$}. $\mathcal{A}$
now sends $\tilde{\mathbf{x}}=\mathbf{Vx}$. Upon receiving a transmission
from user $\mathcal{A}$, user $\mathcal{B}$ computes $\tilde{\mathbf{y}}=\mathbf{U}^{H}\mathbf{y}$.
It is easy to show that this expands to $\mathbf{\tilde{\mathbf{y}}=\Sigma x+\tilde{e}}$.
Since $\mathbf{\Sigma},$ representing the singular values of $\mathbf{A}$,
is a diagonal matrix, $\mathcal{B}$ can efficiently estimate $\mathbf{x}$
with linear complexity in $n$. Notice that $\mathbf{U}$ is unitary
so $\left\Vert \mathbf{e}\right\Vert =\left\Vert \mathbf{\tilde{e}}\right\Vert $.

Now consider the message received by $\mathcal{EVE}$: 
\begin{equation}
\mathbf{\tilde{y}}=\mathbf{BVx}+\mathbf{\tilde{e}}.
\end{equation}

Note that $\mathbf{V}$ consists of the right singular vectors of
$\mathbf{A}$, which is independent of $\mathbf{B}$ and unitary.
Gaussian random matrices are orthogonally invariant \cite{ER05},
so since $\mathbf{V}$ is unitary, multiplying by $\mathbf{V}$ returns
the matrix to an identical, independent distribution. In other words,
the entries of $\mathbf{BV}$ are i.i.d., following the same distribution
as $\mathbf{B}$.

As the main result of this work, we will prove that the computational
complexity for $\mathcal{EVE}$ to efficiently recover $\mathbf{x}$
can be mapped to the problem of solving standard lattice problems
which are conjectured to be computationally hard. 

\subsection{Definitions}

The following definitions will be used in the proofs of our results.
We define a \emph{negligible} amount in $n$, denoted $\mbox{negl.}\left(n\right)$
as any amount being asymptotically smaller than $n^{-c}$ for any
$c>0$. Similarly a \emph{non-negligible} amount is one that is at
least $n^{-c}$ for some $c>0$. We define a \emph{polynomial} amount,
denoted as $\mbox{poly}(n)$, as an amount that is at most $n^{c}$
for some $c>0$. An expression is \emph{exponentially small} in $n$
when it is at most $2^{-\Omega(n)}$, and \emph{exponentially close}
to 1 when it is $1-2^{-\Omega(n)}$. A probability is \emph{overwhelming}
if it holds with probability $1-n^{-c}$ for some $c>0$. An algorithm
is \emph{efficient} or \emph{efficiently computable} if its run time
is within some polynomial factor of the parameter $n$. An algorithm
is \emph{hard} if its run time is at least $2^{n}$. We assume that,
as input, algorithms accept real numbers approximated to within $2^{-n^{c}}$
for some $c>0$. $\mathbb{Z}$ represents the set of all integers
and $\mathbb{Z}_{q}$ represents the set of integers modulo $q$.
Given two probability density functions $\phi_{1}$, $\phi_{2}$ on
$\mathbb{R}^{n}$, we define the \emph{total variational distance
}as 
\begin{equation}
\Delta(\phi_{1},\phi_{2})=\frac{1}{2}\int_{\mathbb{R}^{n}}\left|\phi_{1}(\mathbf{x})-\phi_{2}(\mathbf{x})\right|d\mathbf{x}.
\end{equation}

\subsection{MIMO Signal Distributions }

In this subsection, we define various distributions that are used
in our problem. Specifically, we discuss distributions of lattice
points and distributions which can be empirically related to received
MIMO signals. 

\begin{figure}
\includegraphics[width=1\columnwidth]{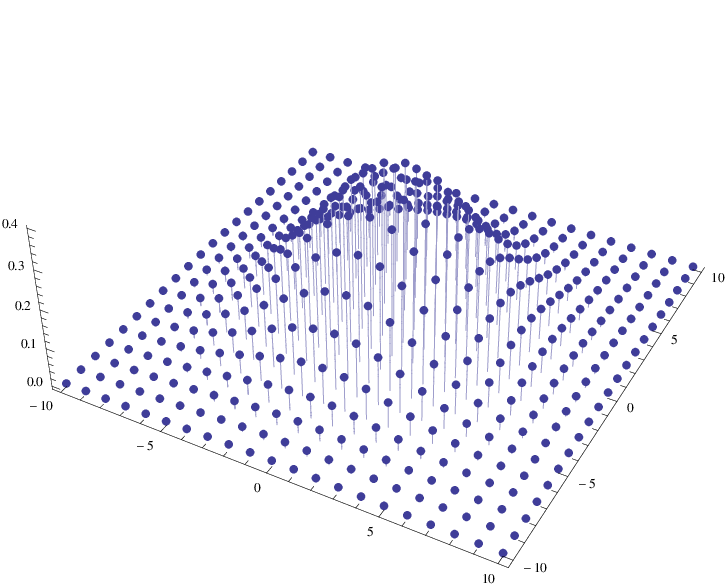}

\caption{A Gaussian distribution on a two-dimensional lattice}
\end{figure}

First, in the proof of our work, we will need to generate lattice
points according to a Gaussian distribution. Since a lattice is a
discrete set of points, we define a \emph{discrete Gaussian distribution},
$D_{A,\alpha}$, for any countable set $A$ and parameter $\alpha$
as 
\begin{equation}
\forall\mathbf{x}\in A,D_{A,\alpha}(\mathbf{x})=\frac{\Psi_{\alpha}(\mathbf{x})}{\Psi_{\alpha}(A)}.
\end{equation}

We now define the distribution, $A_{M,\alpha,k}$ on $\mathbb{R}^{n}\times\mathbb{R}$.
The distribution $A_{M,\alpha,k}$ is the distribution of channel
gains and the received signal from a single antenna in a MIMO system
for a transmitted vector. This distribution is defined for a single
antenna so a MIMO system with $m$ receive antennas would get $m$
samples from this distribution. This distribution is the input to
the MIMO decoding problem which we will define in this section. For
some arbitrary $\mathbf{x}\in\mathcal{X}^{n}$, we define the distribution
$A_{M,\alpha,k}$ as the distribution on $\mathbb{R}^{n}\times\mathbb{R}$
given by drawing $\mathbf{a}$ as defined above, choosing $e_{i}\sim\Psi_{M\alpha}$
and outputting $\left(\mathbf{a}_{i},y_{i}=\left\langle \mathbf{a}_{i},\mathbf{x}\right\rangle +e_{i}\right)$.
We could alternatively express our system as drawing $e_{i}\sim\Psi_{\alpha}$
and outputting $\left(\mathbf{a}_{i},y_{i}=\left\langle \mathbf{a}_{i},\mathbf{x}\right\rangle /M+e_{i}\right)$,
which would result in the same received SNR.

\subsection{The MIMO Decoding Problem}

Given these distributions, we now precisely define MIMO decoding for
the eavesdropper in the MIMO wiretap channel, which we denote as the
\texttt{MIMO-Search} problem. The search problem asks us to recover
the transmitted vector $\mathbf{x}$ without error. In Section V,
we discuss how to use the hardness of the problem to achieve cryptographically
secure systems, and provide a comparison between cryptographer's notions
of security with information theoretic ones. In Appendix A, we prove
that the search problem is as hard as solving certain lattice problems.
We also use the term ``MIMO decoding problem'' to refer to the search
problem.

We wish to show that the MIMO decoding problem, defined below, is
hard to solve, i.e., that this decoding is of exponential complexity
in the number of transmit antennas. We say that an algorithm solves
this problem if it returns the correct answer with a probability greater
than $1-n^{-c}$, for some $c>0$.

\vspace{2mm}

\noindent \textbf{Problem Definition 1.} $\mathtt{MIMO-Search}_{M,\alpha,k}$.
Let $M\geq2$, $\alpha\in(0,1)$, $k\in\mathbb{R}$, $n>0$. Given
a polynomial number of samples of $A_{M,\alpha,k}$, output $\mathbf{x}$.

\vspace{2mm}

The MIMO search problem above will be related to solving standard
lattice problems (discussed in Section III) through standard reductions.
In Section III and Appendix B we will also discuss the complexity
of solving these lattice problems. The result of our reduction implies
that the complexity of decoding a MIMO system grows exponentially
in the number of transmit antennas, i.e., MIMO decoding is at least
as hard as the lattice problems to which we relate them. More precisely,
we state this in a contrapositive manner below. This contrapositive
statement allows us to conjecture a lower bound on the complexity
of solving the MIMO decoding problems, based on the conjectured hardness
of solving lattice problems. 

\vspace{2mm}

\noindent \textbf{Main Result.} Let $M>m\,2^{n\log\log n/\log n}$
for some $m>0$, $\alpha\in\mathbb{R}$, and $k\in\mathbb{R}$ be
such that $m\alpha/k^{2}>\sqrt{n}$. Given access to an efficient
algorithm that can solve $\mathtt{MIMO-Search}_{M,\alpha,k}$, there
exist efficient classical and quantum solutions to standard lattice
problems, which are conjectured to be hard.

\vspace{2mm}

We breifly introduce a second problem, the \texttt{MIMO-Decision}
problem. This problem asks whether or not received samples are transmitted
from a MIMO system (with known channel gain matrix) or are generated
from a Gaussian distribution. Similar decision problems are common
in cryptography, for example, in \cite{R05}, Regev shows that the
decision variant of the LWE problem is hard, and uses this fact to
achieve semantic security \textendash{} effectively hiding information
in a random variable that is uniformly random. Due to the exponential
requirement on the constellation size in our main result, a reduction
from the \texttt{MIMO-Decision} problem to the \texttt{MIMO-Search}
problem is not apparent. It is an open problem to show whether or
not this problem is hard, and if so, how the result could be used
to construct a secure cryptosystem.

\vspace{2mm}

\noindent \textbf{Problem Definition 2.} $\mathtt{\mbox{\ensuremath{\mathtt{MIMO-Decision}}}}_{M,\alpha,k}$.
Let $M\geq2$, $\alpha\in(0,1)$, $k\in\mathbb{R},$ $n>0$. Given
a polynomial number of samples, distinguish between samples of $A_{M,\alpha,k}$
and samples of $\Psi_{M(\alpha+k)}$

\subsection{Assumptions}

In this work we offer a proof of security by showing that breaking
our cryptosystem is at least as hard as solving well-known lattice
problems in the worst case. These problems are conjectured to require
exponential running time. This conjecture and the large body of research
behind it is discussed in the following section. For a more through
treatment of the hardness of these problems, see {[}3{]}\textendash {[}12{]}.
Our proof is also based upon the following assumptions: 
\begin{itemize}
\item We assume that the Gaussian channel noise has sufficiently large power
so that $m\alpha/k^{2}>\sqrt{n}$. If the noise is too small, it is
possible that a subexponential algorithm exists to recover the message
(for motivation of this belief consider the subexponential attack
on the LWE problem with low noise described by \cite{AG}). One possible
way to ensure this requirement is met is to add noise at the transmitter.
This intentionally degrades the communication SNR for a trade-off
of security. A further discussion and characterization of our system
in terms of received SNR is found in Section V. 
\item We assume that the constellation size of the system is large relative
to the number of transmit antennas. Specifically, in order for the
proof of Theorem 1 to hold, we require that $M>m\,2^{n\log\log n/\log n}$.
Unlike the noise requirement, it is unclear that sub-exponential decoding
immediately follows if this bound is not strictly met; it is an open
problem as to whether or not this requirement is needed. 
\item We assume that each entry in the channel gain matrix is independent
and Gaussian. More importantly, we assume that the two receivers of
Bob and Eve, have independent channels. The fact that channels between
different antennas are independent is well justified in most scattering
environments. For example, in a uniform scattering environment, it
has been shown that channels are independent over a distance of $0.4\lambda$,
see e.g. \cite{C68} or \cite{J74}. This independence analysis is
extended to MIMO channels in \cite{MF02} and \cite{ECSRR98}. 
\end{itemize}

\section{Lattices}

We now provide an overview of lattices and lattice-based cryptography.
This section contains all of the concepts used in the cryptography
literature that are required in the proof of our main result. We first
define several problems on lattices that are all conjectured to be
hard to solve and are all used in the proof of our main result to
show that MIMO decoding is at least as hard as solving standard lattice
problems. We next provide a discussion on the complexity of solving
these lattice problems, followed by a discussion on the Learning With
Errors ($\mathtt{LWE}$) problem. The Learning With Errors problem
has a striking similarity to the problem of MIMO decoding. We follow
a very similar approach to show the hardness of MIMO decoding as is
used to show the hardness of LWE decoding. 

Lattice-based cryptography has generated much research interest in
recent years. Cryptographers' interest in lattices largely began with
the surprising result of Ajtai \cite{A96} which created a connection
between the average-case and worst-case complexity of lattice problems.
Cryptographers have used these results to create a wide variety of
cryptographic constructions which enjoy strong proofs of security.
Important to this paper is the work of Micciancio and Regev \cite{MR04}
which explores Gaussian measures on lattices. The Learning With Errors
problem, first introduced by Regev \cite{R05}, shows that recovering
a point that is perturbed by a small Gaussian amount from a random
integer lattice can be related to approximating solutions to standard
lattice problems. This problem is discussed more thoroughly below.
The hardness reduction in this work very closely follows the work
in \cite{R05}. For a survey on lattice-based cryptography, see \cite{MR09}.

A \emph{lattice} is a discrete periodic subgroup of $\mathbb{R}^{n}$
that is closed under addition. Alternatively, a lattice can be defined
as the set of all integer combinations of $n$ linearly independent
vectors, known as a \emph{basis}, 
\begin{equation}
\mathcal{L}(\mathbf{A})=\left\{ \mathbf{Ax}=\sum_{i=1}^{n}x_{i}\mathbf{A}_{i}:\mathbf{x}\in\mathbb{Z}^{n}\right\} .
\end{equation}

A lattice is not defined by a unique set of basis vectors. Any basis
$\mathbf{A}$ multiplied by a unimodal matrix results in an alternative
basis representation of the same lattice. We can define a set of lengths,
in the $l_{2}$-sense, $\lambda_{i}(\mathbf{A}),\:\mathrm{{for}}\,i\in\{1,\ldots,n\}$,
as the radius of the smallest ball around the origin that contains
$i$ linearly independent vectors. These lengths are known as the
\emph{successive minima}. It is easy to show that there exists vector
$\|v_{i}\|=\lambda_{i}(\mathbf{A})$ for all $1\leq i\leq n$ and
that $\lambda_{1}(\mathbf{A})\geq\min_{i}\|\mathbf{\tilde{A}}_{i}\|$,
where $\mathbf{\tilde{A}}$ is the Gram-Schmidt orthonormalization
of the basis $\mathbf{A}$, where the norm is taken as the $l_{2}$
norm.

The dual of a lattice $\mathcal{L}(\mathbf{A})\in\mathbb{R}^{n}$,
$\mathcal{L}^{*}(\mathbf{A})$, is the lattice given by all $\mathbf{y}\in\mathbb{R}^{n}$
such that for every $\mathbf{x}\in\mathcal{L}(\mathbf{A})$, $\left\langle \mathbf{x},\mathbf{y}\right\rangle \in\mathbb{Z}$.
Since $\mathbf{A}$ is full rank, $\mathcal{L}((\mathbf{A}^{T})^{-1})$
is the dual of the lattice $\mathcal{L}(\mathbf{A})$.

\subsection{Lattice Problems.}

This section considers a set of standard problems on lattices that
are conjectured to be computationally hard. These problems will be
used in the characterization of the computational complexity of MIMO
decoding. These problems are defined with an approximation factor,
$\gamma(n)\geq1$, and input given in the form of an arbitrary basis.
The precise definition of the approximation factor varies for each
problem, but is precisely stated in each definition below. The approximation
factor plays an important role in the computational complexity required
to solve a given problem. In general, for very large approximation
factors, the following lattice problems can be efficiently computed.
For small approximation factors, the problems are conjectured to require
exponential running time. The problem of MIMO decoding will be related
to solving lattice problems with an approximation factor of $\gamma=n/\alpha$,
which is conjectured to be hard. The relation between the approximation
factor and the complexity of solving the problem is discussed in detail
in Appendix B. For a survey on the hardness of these problems, see
\cite{MR09} or \cite{MG02}.

The first problem we describe is the shortest vector problem. In general,
these problems can be efficiently solved if $\gamma$ is at least
$O\left(2^{n}\right)$, but are conjectured to require exponential
run time for exact solutions or even for polynomial approximation
factors. A close variant of this problem, the decision Shortest Vector
Problem ($\mathtt{GapSVP}$), asks whether or not the shortest vector
generated by the lattice is shorter than some distance $d$. In general,
for an arbitrary basis, short vectors are hard to find. 
\begin{defn}
$\mathtt{GapSVP_{\gamma}}$ is defined as follows: given an $n$-dimensional
lattice $\mathcal{L}(\mathbf{A})$ and a number $d>0$, output YES
if $\lambda_{1}(\mathbf{A})\leq d$ or NO if $\lambda_{1}(\mathbf{A})>\gamma(n)\cdot d$.
\end{defn}
The following problem, the shortest independent vectors problem ($\mathtt{SIVP}$),
is often referred to as a lattice basis reduction.
\begin{defn}
$\mathtt{SIVP}_{\gamma}$ is defined as follows: given an $n$-dimensional
lattice $\mathcal{L}(\mathbf{A})$, output a set of $n$ linearly
independent vectors of length at most $\gamma(n)\cdot\lambda_{n}(\mathbf{A})$. 
\end{defn}
One well-known algorithm for finding approximate solutions to the
$\mathtt{SIVP}$ problem is the Lentra-Lentra-Lovász (LLL) lattice
basis reduction algorithm, which creates an ``LLL-reduced'' basis
in polynomial time. For large values of $n$, the LLL algorithm returns
a basis that is exponentially larger than the shortest possible basis
of the lattice. It is generally conjectured that no polynomial time
algorithm could approximate $\mathtt{SIVP}$ to within a polynomial
factor of $n$. For a more in-depth discussion of lattice basis reduction
algorithms and their complexity, see \cite{MR09}.

The following two lattice problems are important in the reduction
given in this paper. They are also used by Regev in \cite{R05} in
his proof of the security of LWE. The first problem, the Bounded Distance
Decoding ($\mathtt{BDD}$) problem, is equivalent to decoding a linear
code, where the received vector is at most some distance $d$ from
the nearest codeword. Exact solutions to this problem are conjectured
to require exponential run time. By requiring $d<\lambda_{1}/2$,
we ensure that the closest vector is unique. The Closest Vector Problem
(CVP) is identical to this problem but without a bounding distance.
The closest vector problem arrises in many other contexts, for example
it is equivalent to decoding a linear code, and also the maximum-likelihood
decoding problem in MIMO systems.
\begin{defn}
The $\mathtt{BDD}_{\mathcal{L}(\mathbf{A}),d}$ problem is defined
as follows: given an $n$-dimensional lattice $\mathcal{L}(\mathbf{A})$,
a distance $0<d<\lambda_{1}(\mathbf{A})/2$, and a point $\mathbf{x}\in\mathbb{R}^{n}$
which is at most distance $d$ from a point in $\mathcal{L}(\mathbf{A})$,
output the unique closest vector in $\mathcal{L}(\mathbf{A})$ 
\end{defn}
Informally, we can relate MIMO decoding to $\mathtt{BDD}$ (or more
generally to the Closest Vector Problem), as follows. Linear precoding
is well known to simplify encoding and decoding of a MIMO system to
be of polynomial complexity in $n$. In terms of lattice problems,
we say that this linear precoding transforms the lattice basis into
one in which the Closest Vector Problem is easy to solve. This is
very closely related to the cryptographic notion of the trapdoor function:
a function that is easy to compute in one direction, but computationally
infeasible to invert without a key (which serves as the ``trapdoor'')
\cite{DH}. The linear precoding transformation applied to the MIMO
channel effectively creates a spatially-varying trapdoor in that it
allows the MIMO channel to be efficiently inverted only at Bob's location.
In our model, both communicators must have knowledge of the channel
in order to perform the parallel decomposition via a singular value
decomposition (SVD) associated with linear precoding. The eavesdropper
may also learn this channel and hence learn the ``key'', but this
does not allow it to decode the message with lower complexity. 

We finally define the discrete Gaussian sampling problem ($\mathtt{DGS}$),
which, for small values of $r$, can be reduced to solving $\mathtt{GapSVP}$
and $\mathtt{SIVP}$. See \cite{R05} for this reduction. This problem
asks us to sample a Gaussian distribution, with standard deviation
$r/\sqrt{2\pi}$, with support over a lattice. 
\begin{defn}
The $\mathtt{DGS}{}_{r}$ problem is defined as follows: given an
$n$-dimensional lattice $\mathcal{L}(\mathbf{A})$ and a number $r$,
output a sample from $D_{L,r}$, the discrete distribution defined
in Section II.3. 
\end{defn}
We borrow the following claim from \cite{GPV08}, which gives an efficient
algorithm to solve $\mathtt{DGS}{}_{r}$ for large $r$. Intuitively,
solving $\mathtt{DGS}{}_{r}$ for small values of $r$ becomes a computationally
harder task because it reveals information about the short vectors
in the lattice. We use the following claim in the proof of our main
theorem, as our main theorem requires sampling Gaussian distributions
on lattices. 
\begin{claim}
\cite[Theorem 4.1]{GPV08}. There exists a probabilistic polynomial-time
algorithm that, given an $n$-dimensional lattice $\mathcal{L}(\mathbf{A})$,
and a number $r>\lambda_{n}(\mathcal{L}(\mathbf{A}))\cdot\omega\left(\sqrt{\log n}\right)$,
outputs a sample from a distribution that is within negligible distance
of $D_{\mathcal{L}(\mathbf{A}),r}$. 
\end{claim}

\subsection{Learning With Errors}

We now give an overview of the Learning With Errors problem. The reader
may note the similarities between this problem and the problem of
MIMO decoding. In particular, the Learning with Errors problem is
very similar to the problem presented in this paper except that Learning
with Errors is set entirely over integer fields, whereas MIMO decoding
is set in the reals. We take advantage of the results related to the
Learning with Errors problem to show the hardness of MIMO decoding.

The Learning with Errors ($\mathtt{LWE}$) problem was first presented
by Regev in \cite{R05}. The problem allows one to have an arbitrary
number of samples of ``noisy random inner products'' of the form
\begin{equation}
\left(\mathbf{a},y=\left\langle \mathbf{a},\mathbf{x}\right\rangle +e\right),
\end{equation}
where $\mathbf{a}\in\mathbb{Z}_{q}^{n}$, $\mathbf{x}\in\mathbb{Z}_{q}^{n}$,
and $e\in\mathbb{Z}_{q}$. Each $\mathbf{a}$ is random from the uniform
distribution over $\mathbb{Z}_{q}^{n}$, and each $e$ is drawn from
a discrete Gaussian distribution with a small second moment (relative
to $q$). This problem is an extension of the classic \emph{learning-parity
with noise }problem of machine learning. As with many hard problems,
the $\mathtt{LWE}$ problem has two variants: the search and decision
variants. These two problems are related in complexity. In the search
variant of the problem, the goal is to recover the vector $\mathbf{x}$;
whereas the decision variant asks us to distinguish between samples
of $y$ and a random distribution. If we fix the number of samples
to which one has access, the $\mathtt{LWE}$ problem becomes analogous
to finding the closest vector in an integer lattice. In \cite{R05},
Regev proves that the $\mathtt{LWE}$ problem is at least as hard
as solving $\mathtt{GapSVP}_{n/c}$ and $\mathtt{SIVP}_{n/c}$ in
the worst case, but requires the use of a quantum computer. Peikert
in \cite{P09} demonstrates a reduction that does not require a quantum
computer for the case where $q=O\left(2^{n}\right)$.

If we fix the number of samples available to the receiver, then this
problem becomes equivalent to decoding a random linear code. Consider
the case where we take a random generator matrix $\mathbf{A}\in\mathbb{Z}_{q}^{m\times n}$,
a vector $\mathbf{t}=\mathbf{Ax}+\mathbf{e}\in\mbox{\ensuremath{\mathbb{Z}}}_{q}^{m}$
and we wish to recover the vector $\mathbf{x}$. If we restrict $m<\mbox{poly}(n)$,
and provide the proper distribution on $\mathbf{e}$, then decoding
this code is as hard as solving worst-case lattice problems to within
a factor of $\gamma<n/c$.

We briefly summarize Regev's reduction as follows. If an efficient
algorithm exists to solve $\mathtt{LWE}$, then this can be used to
construct an efficient algorithm that solves the $\mathtt{BDD}$ problem
for any $n$-dimensional lattice. This step is somewhat surprising
as it allows us to use the $\mathtt{LWE}$ algorithm, which operates
on an integer lattice reduced modulo $q$, to solve $\mathtt{BDD}$
for any real lattice. How this step is accomplished is not entirely
intuitive, but requires one to convert a single instance of $\mathtt{BDD}$
into an arbitrary number of samples of the distribution defined in
the $\mathtt{LWE}$ problem. This step also requires as an input samples
of a number of lattice points drawn from a discrete Gaussian distribution
with a large second moment. With an efficient algorithm to solve $\mathtt{BDD}$,
it is possible to generate a distribution of lattice points which
has a smaller second moment than input in the previous step. The above
process can now be iterated with this new distribution of lattice
points to solve $\mathtt{BDD}$ with a larger bounding distance than
before. Repeating this step will eventually result in a distribution
of lattice points with a very small second moment, which will reveal
information about the shortest vectors of the lattice, allowing us
to approximate solutions to $\mathtt{GapSVP}$ and $\mathtt{SIVP}$.

This step, using an $\mathtt{LWE}$ algorithm to solve any instance
of the $\mathtt{BDD}$ problem, turns out to be extremely useful to
characterize the hardness of MIMO decoding. In particular, we will
show that if we have an efficient algorithm for MIMO decoding, then
this algorithm can be used to efficiently to solve instances of the
$\mathtt{BDD}$ problem. Then, with an efficient $\mathtt{BDD}$ algorithm,
Regev's quantum reductions to worst-case lattice problems follow.
Additionally for large $M$, Peikert's classical reduction follows. 

\subsection{Lattice-Based Cryptography}

In recent years, lattice-based cryptography has become a very attractive
field for cryptographers for a number of reasons. The security guarantees
provided by many lattice-based schemes far exceeds that of many modern
schemes such as RSA and Diffie-Hellman, since lattice problems enjoy
an average-to-worst case connection, as discussed in Appendix B. Lattice
problems also appear to be resistant to quantum computers. The creation
of such computers poses a significant challenge to the state of modern
cryptography, since a quantum computer could break most modern number-theoretic
schemes. In addition, many proposed lattice schemes are competitive
with, or better than, many modern number-theoretic schemes in terms
of both key size and the computational efficiency of encryption and
decryption \cite{BB13}. For these reasons, a number of lattice-based
cryptography standards are being developed or have been developed
recently by both the IEEE and the financial industry (\cite{IEEE_LS},
\cite{ANSI}).

Lattice-based cryptography provides cryptographers with a wide variety
of tools to create many different cryptographic constructions. We
here reference some of these constructions as it is possible that
some of them could be applied to the MIMO decoding problem or even
that the MIMO decoding construction could inspire entirely new cryptographic
constructions. Importantly, the $\mathtt{LWE}$ problem can be extended
to cyclotomic rings (\cite{LPR}), which can improve the efficiency
(in terms of key size and complexity of the encryption and decryption
operations) of cryptosystems based on $\mathtt{LWE}$. In \cite{R05},
a system that is secure against Chosen-Plaintext Attack (CPA-secure)
is presented and in \cite{P09} a Chosen-Ciphertext Attack secure
(CCA-secure) system is presented. In addition, the $\mathtt{LWE}$
problem has been used to create a wide range of useful cryptographic
primitives such as identity-based encryption (\cite{BGB05}), oblivious
transfer (\cite{PVW08}), zero-knowledge systems (\cite{MV03}), and
pseudorandom functions (\cite{BPR11}). The $\mathtt{LWE}$ problem
was also used to construct a fully homomorphic cryptosystem \cite{FHE},
solving a long-standing open problem in cryptography.

\section{Main Theorem}

\begin{figure}
\begin{centering}
\includegraphics[width=1\columnwidth]{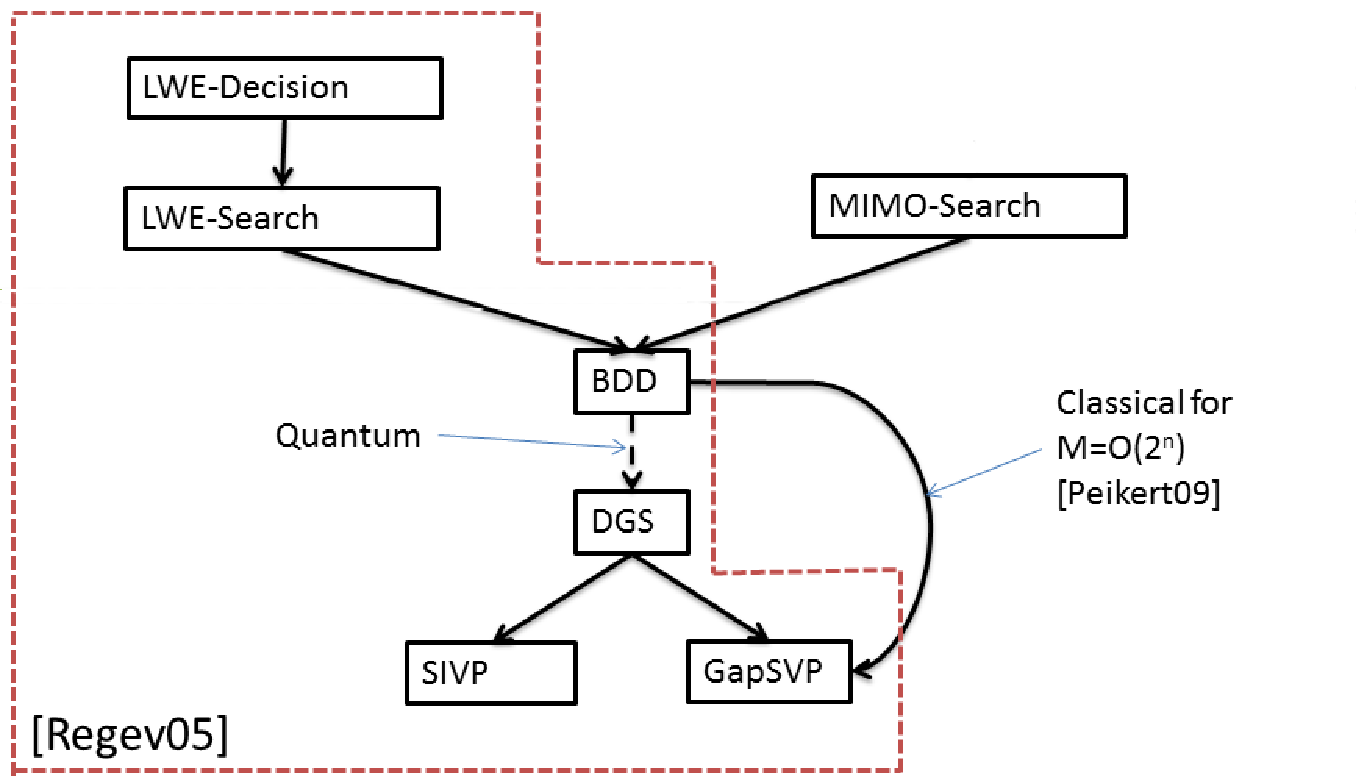} 
\par\end{centering}
\caption{A map of reductions relating MIMO decoding (the MIMO-search problem)
to solving standard lattice problems and the LWE problem. If there
exists an efficient algorithm to solve the MIMO-Search problem, then
this implies solutions to standard lattice problems. Since lattice
problems are conjectured to be hard, this conjecture follows for the
hardness of the MIMO-search problem. In this figure, $M$ refers to
the constellation size used in the MIMO system.}
\end{figure}

In this section we state our main theorem regarding the computational
hardness of $\mathtt{\mbox{\ensuremath{\mathtt{MIMO-Search}}}}_{M,\alpha,k}$.
We show that given an efficient algorithm which can solve this problem,
there exist efficient solutions to the problems of $\mathtt{GapSVP}_{n/\alpha}$
and $\mathtt{SIVP}_{n/\alpha}$. The reduction is based off the work
of \cite{R05} and involves a step using a quantum computer. In \cite{P09}
this step is removed, but the hardness is only related to the $\mathtt{GapSVP}$
problem. The relation between the reductions used to show the hardness
of MIMO decoding and the reductions used to show the hardness of $\mathtt{LWE}$
is shown in Figure 3. Examples of parameters which meet the requirements
stated in Theorem 1 are shown in Table I. The proof of this theorem
is found in Appendix A.

\noindent \textbf{Theorem 1.} \emph{$\mathtt{MIMO-Search}_{M,\alpha,k}$
to $\mathtt{GapSVP}_{n/\alpha}$ and $\mathtt{SIVP}_{n/\alpha}$.
Let} $m>0$, $\alpha\in\mathbb{R}$, $k\in\mathbb{R}$,\emph{ $M>m\,2^{n\log\log n/\log n}$,
be such that $m\alpha/k^{2}>\sqrt{n}$. Assume we have an efficient
algorithm that solves $\mathtt{\mbox{\ensuremath{\mathtt{MIMO-Search}}}}_{M,\alpha,k}$,
given a polynomial number of samples from $A_{x,\alpha,k}$. Then
there exists an efficient quantum algorithm that, given an $n$-dimensional
lattice $\mathcal{L}(\mathbf{A})$, solves the problems $\mathtt{GapSVP}_{n/\alpha}$
and $\mathtt{SIVP}_{n/\alpha}$. Hence, since $\mathtt{GapSVP}_{n/\alpha}$
and $\mathtt{SIVP}_{n/\alpha}$ are conjectured to be hard, it is
also conjectured that MIMO-Search is hard.}

An outline of the reductions is shown in Figure 3. Figure 3 also relates
the work of Regev and Peikert to our work. The steps used to prove
the theorem are outlined as follows: 
\begin{itemize}
\item We first show that, given the MIMO oracle as described in Section
II, we can solve problems where the coefficients of the channel gain
matrix are instead drawn from a discrete Gaussian distribution as
described in Lemma 1 in Appendix A. 
\item We begin by reducing a lattice basis \textbf{$\mathbf{A}$} by using
the Lensta-Lensta-Lovász (LLL) lattice-basis reduction algorithm.
We then, using the procedure described in \cite{GPV08}, create a
discrete Gaussian distribution on this lattice, with a second moment
around the length of the largest vector given in the reduced basis.
We use this as the starting point for the iterative portion of the
algorithm. 
\item The main step in the proof is given in Lemma 7, where we use this
MIMO decoding oracle to solve the $\mathtt{BDD}$ problem given access
to a $\mathtt{DGS}$ oracle. This allows us to directly use the results
from Regev \cite{R05} and Peikert \cite{P09}. 
\item In Lemma 8, from \cite{R05}, the $\mathtt{BDD}$ oracle is used to
(quantumly) solve $\mathtt{DGS}_{L^{*},\sqrt{n}/\left(\sqrt{2}d\right)}$,
that is return samples of $D_{L^{*},r}$. Note that we can efficiently
sample from $D_{L,r}$ for $r>\eta_{\epsilon}(L)$. If in Lemma 7
we set parameters so that $\sqrt{2}d>\sqrt{n}$, then we can reduce
the value of $r$ to below the value for which we could previously
efficiently sample, that is we can construct a distribution that is
more narrow than previously possible.
\item The $\mathtt{BDD}$ and $\mathtt{DGS}$ oracles can now be applied
iteratively, shrinking the second moment of the discrete Gaussian
distribution with each iteration. Eventually, the distribution becomes
narrow enough to reveal information about the shortest vectors of
the lattice, thereby solving the \emph{$\mathtt{GapSVP}_{n/\alpha}$
}and\emph{ $\mathtt{SIVP}_{n/\alpha}$ }problems\emph{.}
\end{itemize}
\begin{table}
\begin{centering}
\caption{Maximum SNRs}
\begin{tabular}{|c|c|c|}
\hline 
$n$ & $\log_{2}M$ & SNR (dB)\tabularnewline
\hline 
80  & 33.7 & 87.1\tabularnewline
\hline 
128  & 51.3 & 139.2 \tabularnewline
\hline 
196  & 75.4 & 210.7 \tabularnewline
\hline 
256 & 96 & 272.2\tabularnewline
\hline 
\end{tabular}
\par\end{centering}
\smallskip{}

Example sets of parameters for various numbers of transmit antennas.
In order for the security proofs contained in this paper, the system
must meet these minimum constellation size and maximum SNR values.
Here we set $m=1$.
\end{table}

\section{Security and the MIMO Wiretap Channel}

\noindent In the previous sections we have shown that the MIMO wiretap
model is secure in the sense that the decoding complexity of any eavesdropper
is exponential in the number of transmit antennas. In this section,
we relate this notion of security for our system to other notions
of security used by information theorists and cryptographers. 

Until this section, we have referred to the third user on Alice and
Bob's public channel to be an eavesdropper. This implies that the
eavesdropper is passive and has no ability but to receive and decode
information. In practice, this model is limited and, in order to provide
more practical and robust notions of security, a more powerful adversarial
model should be considered. The need for using more robust security
models in the study of physical-layer security was recently discussed
in \cite{Trappe}. In this work the author suggests bridging the gap
between notions of security in information theory and cryptography
in order to make physical-layer schemes more widely accepted by the
security community. Considering an adversary which, for example, has
the ability to manipulate, inject, alter or duplicate information
is considered essential when a cryptographer develops and designs
a cryptosystem, but is rarely, if ever, considered in information-theoretic
settings. In this section, we discuss stronger adversarial models
which are commonly considered in cryptography and show how to use
the hardness result derived in Section IV in order to construct secure
schemes under these models. We note that there are many other adversarial
models studied in the field of cryptography which we do not discuss
here. A more complete discussion of these models is provided in \cite{KL07}.

In this section, we present two ways in which Alice and Bob can securely
use our result. In the first scheme, we show how Alice may securely
transmit a message to Bob. This scheme achieves both the notion of
Chosen Plaintext Security (also known as Semantic Security) as well
as Chosen Ciphertext Security. However, the downside of this scheme
is that it relies on additional cryptographic primitives and assumptions:
this scheme is secure in what cryptographers call the Random Oracle
Model, which we describe in Section V.B. In Section V.C, we present
a simple scheme which allows Alice and Bob to agree on a secret key
which, in order to be practical, must take on additional assumptions
about Eve's ability to approximately decode. Constructing a cryptosystem
which is provably secure using only the assumptions listed in Section
II is not readily apparent. We believe that the schemes proposed are
practical and that applying these notions of security in the context
of physical-layer security is not only novel, but makes substantial
progress in bridging ideas from information-theoretic security and
cryptography, as discussed in \cite{Trappe}. 

\subsection{Information-Theoretic Security}

We have shown that it is hard for Eve to decode Alice's message. However,
we defined decoding as exactly recovering the message without error.
This is not the same as stating that, information-theoretically, Eve
receives no information about the message. Consider that Eve applies
some suboptimal decoder that runs in polynomial time (for example,
successive interference cancellation or zero-forcing decoding). Eve
would then get some estimate for Alice's message with some bit error
rate that is potentially less than half. In some sense, this is insecure
as it implies that Eve in fact receives information about Alice's
message. In this section, we provide a means to use our system in
a manner that is in fact secure under cryptographic notions of security
rather than information-theoretic ones.

Under the widely accepted conjecture that lattice problems cannot
efficiently be approximated to even subexponential factors, Eve can
at best recover a constellation point that is exponentially far away
from Alice's transmitted vector. Effectively, this means that for
a single channel use, Eve can reduce the number of possible transmitted
vectors from $M^{n}$ to $2^{O(n)}$. 

We could use this result in an information-theoretic sense along with
our hardness result to say that there is now a positive secrecy rate
in the wiretap channel, which could be accomplished through either
a key rate, coding, or using a transmission strategy that takes into
account the channel characteristics of Bob and/or Eve (e.g. puts a
null in the direction of Eve). While this approach would extend existing
results on information-theoretic secrecy, it is different from our
approach in two main aspects, as we now describe.

First, assume that the average received SNR at Bob's location is identical
to the average received SNR at Eve's location and, at both locations,
the sufficient conditions of our hardness result holds. Since the
channel gain matrices between Alice and Bob and Alice and Eve have
the same distribution, on average they have the same number of spatial
streams that can be supported and at the same SNR. In other words,
the mutual information between the channel input and the output at
Bob and at Eve is the same, unless the input is encoded in a manner
that allows Bob to decode whereas Eve cannot. This will reduce the
mutual information between Alice and Eve, possibly to zero. On the
other hand, our approach shows that there is (likely) no efficient
algorithm that allows Eve to decode, and thus, \emph{practically},
if Eve is subject to complexity constraints she will receive less
information than Bob. To our knowledge, this notion of information\emph{
}exchange\emph{,} based on the conjectured complexity of decoding,
is novel.

Second, relying on the information theoretic secrecy argument does
not provide a constructive result. By this we mean that it merely
shows the existence of a scheme that allows for Alice and Bob to reliably
and securely communicate, rather than actually providing such a scheme,
as is the case for most information-theoretic results. In contrast,
the cryptographic notion of security that we use in fact gives a practical
scheme for secure communication.

\subsection{Secure Message Transmission in the Random Oracle Model}

We have shown that it is computationally infeasible for Eve to exactly
recover Alice's transmitted message, but Bob can decode in polynomial
time. This implies that (assuming lattice problems are hard) the MIMO
channel naturally forms what cryptographers call a secure one-way
trapdoor function\footnote{A one-way function is a function which can be evaluated in polynomial
time, but requires exponential complexity to invert. Such functions
are conjectured to exist, but proving their existence is an open problem
and would prove $\mathtt{{P}}\neq\mathtt{{NP}}$.}. In this subsection, we construct a secure cryptosystem by using
a standard cryptographic construction, known as OAEP+ (Optimal Asymmetric
Encryption Padding). OAEP+ was introduced in \cite{OAEPP} as an improvement
to the original OAEP scheme given in \cite{OAEP}. The results of
\cite{OAEPP} allow for the construction of a secure cryptographic
system given any one-way permutation. The results imply that, in the
random oracle model, the resulting system has the properties of Chosen
Ciphertext (CCA) Security and Chosen Plaintext (CPA) Security. Precisely,
the security of our system follows from the following theorem, which
is proven in \cite{OAEPP} for any one way permutation. The fact that
the MIMO-Search problem is one way, under the assumption that lattice
problems are hard, is given by Theorem 2\footnote{The MIMO channel maps messages from bits to real values and thus is
a function and not a permutation. However, the MIMO search problem
asks Eve to decode a MIMO signal back to bits, and thus we can view
the process of transmitting, detecting and decoding as a permutation
of the message bits.}. Theorem 2 gives strong evidence that at least one bit (and likely
more) will remain computationally hidden from Eve, and thus Eve will
learn no information about the transmitted message.

\setcounter{thm}{1}
\begin{thm}
\cite[Thm. 3]{SecretKeyAgreement} If the underlying trapdoor permutation
is one way, then OAEP+ is secure against adaptive chosen ciphertext
attack in the random oracle model.
\end{thm}
We now give an informal treatment of CPA and CCA Security, as well
as the random oracle model. For a formal treatment, see \cite{KL07}
or \cite{BellEq}. A reader familiar with these notions may wish to
skip ahead to Algorithm 1. 

Succinctly, the condition necessary to achieve security under Chosen
Plaintext Attack (CPA security \textendash{} here we describe the
notion referred to as \texttt{IND-CPA}) is that an adversary must
be unable (without violating our computational hardness assumptions)
to match up pairs of plaintext (unencrypted data) and ciphertext (encrypted
data). This notion of CPA security was formally introduced in \cite{MG84}.
There are many reasons why such a notion of security is needed. For
example, consider that an adversary knows that a user communicates
with only a known small subset of possible messages. If the adversary
can match plaintext and ciphertext, then the adversary can exhaustively
search over the possible transmitted messages and learn the message. 

We illustrate the value of this notion of security with an example
from our model in Figure 1. Sending non-random messages, or messages
from a source with a small entropy, could be problematic from a security
perspective. Consider, for example, that Alice is casting a vote in
an election which she sends to the election official, Bob. In the
election, there are a small number, $v$, of candidates. Alice votes
by sending some message $\mathbf{x}_{i}$, for $i\in[0,v)$, which
is predefined, and publicly known, for each candidate. Eve could receive
$\mathbf{y}=\mathbf{B}\mathbf{x}_{i}+\mathbf{e}$, and compute the
$l_{2}$ norm between $\mathbf{y}$ and each possible $\left\{ \mathbf{Bx}_{0},\ldots,\mathbf{Bx}_{v-1}\right\} $
and tell with some certainty which candidate received Alice's vote.
If instead Alice transmits her vote using a scheme that is CPA secure,
receiving the vector $\mathbf{y}$ will convey no information about
Alice's vote to Eve.

Security under Chosen Ciphertext Attack (CCA) is a stronger notion
than CPA security (and in fact CCA security, as described here, implies
CPA security, see \cite{BellEq}). Under the CCA adversarial model,
an adversary first gets access to a decryption oracle (but not the
secret key) and is allowed to present his choice of ciphertext to
the oracle and learn the corresponding plaintext. Then after loosing
access to this oracle, the adversary is then challenged to match pairs
of plaintext and ciphertext. This is also referred to as a ``lunchtime'',
non-adaptive CCA security, or \texttt{IND-CCA1}. In a stronger notion
of CCA security (known as adaptive CCA or \texttt{IND-CCA2}), the
adversary can continue to submit queries to the oracle after being
presented with the challenge plaintext and ciphertext pairs, under
the condition that he does not submit the challenge ciphertext to
the oracle. 

The necessity for CCA security in the context of our system is not
as readily apparent as it is in a traditional keyed cryptosystem.
Suppose, for example, a device exists which decrypts ciphertext, but
from which a party cannot recover the secret key. If an adversary
is given access to such a device, then a CCA secure cryptosystem is
still secure given that this access is temporary. In our system model
however, such a device cannot exist: the ability for Bob to decode
Alice's messages depends on both Alice and Bob to be at certain spatial
locations and not on the possession of a key. In fact, when Eve receives
a message through the wiretap channel, she receives the vector $\mathbf{BVx+e}$,
which was precoded for Bob's channel. There is no way for \emph{anyone}
to efficiently decode this vector. Nonetheless, our system is still
secure if such an oracle were to exist, as adaptive CCA security (and
hence CPA security) follows from using the OAEP+ construction.

The Random Oracle model was introduced by \cite{BellRO}. The model
assumes that parties have access to an oracle which choses uniformly
at random a function which maps each possible input query into its
output domain. When implemented in a practical scheme, such an oracle
can be replaced by a cryptographic hash function, such as the SHA
family of hashes. For many cryptographic hash functions, proving certain
security properties is a difficult task. The Random Oracle model provides
cryptographers an idealized model of a hash function to make possible
rigorous analysis of cryptographic constructions that include such
hash functions. Constructions which are provably secure under the
Random Oracle model have been standardized and their use is extremely
common.

\begin{algorithm}[h]
\caption{MIMO-OAEP+}

Let $K=n\log M$ be the number of bits transmitted per MIMO channel
use. Alice wishes to send Bob a message, \textbf{$\mathbf{m}$}, that
is $\eta=K-2n$ bits. Assume Alice and Bob both have access to three
random oracle functions: $G:\left\{ 0,1\right\} ^{n}\rightarrow\left\{ 0,1\right\} ^{\eta},\,H':\left\{ 0,1\right\} ^{n+\eta}\rightarrow\left\{ 0,1\right\} ^{n},\,H:\left\{ 0,1\right\} ^{n+\eta}\rightarrow\left\{ 0,1\right\} ^{n}$.
Alice computes $s\in\left\{ 0,1\right\} ^{n+\eta},\:t\in\left\{ 0,1\right\} ^{n},\,x\in\left\{ 0,1\right\} ^{K}$
as shown under \texttt{Encrypt} below. Alice then multiplies $\mathbf{x}$
by the right singular vectors of Bob's channel and transmits the message
to Bob. Bob recovers $\mathbf{x}$ through his channel and then recovers
$\mathbf{m}$ using the proceedure shown under \texttt{Decrypt}. Bob
verifies that $c=H'(r\|m)$. If these quantities are not equal then
Bob has not properly received Alice's message and he rejects.

\begin{align*}
    &\omit\texttt{Encrypt}                        &   & \omit\texttt{Decrypt} \\    
s   &= \left(G(r) \oplus m \right) \| H'(r \| m)   & s &= x[0,\ldots,\eta+n-1]         \\
t   &= H(s) \oplus r                               & t &= x[\eta+n,\ldots,k]           \\
x   &= s \| t									  & r &= H(s) \oplus t               \\
	&											  & m &= G(r) \oplus s[0,\ldots,\eta-1] \\
    &                                              & c &= s[\eta\ldots, \eta+n - 1]       \\
    &                                              & c &\stackrel{?}{=} H'(r\|m)      
\end{align*}
\end{algorithm}

The proposed scheme is given in Algorithm 1, which simply encodes
a message using the OAEP+ construction given in \cite{OAEPP} before
transmitting the message over the MIMO channel. As a result of this
encoding, the decoder must receive the message $\mathbf{x}$ without
error in order to decode. If the decoder has even a single error,
then the resulting output will be computationally indistinguishable
from uniform random by Theorem 2 and the one-wayness of the MIMO channel. 

We note that this algorithm requires Bob to receive a perfect copy
of the message $\mathbf{x}$ through the MIMO channel. However, Bob
receives this message through a channel with Gaussian noise and thus
has some probability of error that he does not receive this message
perfectly. In fact, given the large noise requirement in Theorem 2,
it can be shown that the symbol error rate at Bob's location will
be close to one, even when Bob has a large number of antennas.

In order to allow Alice and Bob to communicate, we must reduce the
rate of communication between them; however, in order to allow for
Theorem 2 to hold, we must still maintain the large constellation
size and noise requirement. Both of these objectives can be accomplished
through the use of error correcting codes. Note, however, that if
we encode the message $\mathbf{x}$ with a code that can correct up
to $\mathbf{e}$ errors, then this also aids Eve in recovering the
plaintext. Because Eve cannot apply ML decoding, she will experience
an increased error rate over Bob. We can use this fact to construct
a code which Bob can decode but Eve cannot.

\subsubsection{Eavesdropper Decoding}

We now consider the limits to which Eve can estimate the message $\mathbf{x}$
which is not precoded for her channel. The analysis in this section
describes how much Eve's estimate of $\mathbf{x}$ differs from its
maximum-likelihood, which we denote as $\hat{\mathbf{x}}$. We obtain
a lower bound on this distortion by considering state-of-the-art algorithms
which are used to attack lattice-based cryptography schemes. We begin
by stating a conjecture about the limits of approximating lattice
algorithms. This conjecture is widely accepted to be true (see \cite{MR09},
for example) and discussed in more detail in Appendix B:
\begin{conjecture}
There is no polynomial time algorithm which can approximate $\mathtt{SIVP}$
to within polynomial factors.
\end{conjecture}
Alice transmits a message $\mathbf{x}\in[0,M)^{n}$, which is distributed
uniformly, across her channel. Eve receives this message through her
channel and can apply her favorite lattice-basis reduction algorithm
to decoding the message, but is unable to obtain a lattice basis that
is within a polynomial factor of the shortest basis. She can use this
shorter basis along with Babai's nearest plane algorithm \cite{Babai}
to recover a vector, $\mathbf{x}'$, that is a super-polynomial distance
away from $\hat{\mathbf{x}}$, that is $\|\mathbf{x'}-\hat{\mathbf{x}}\|>2^{\omega(\log n)}$.
The fact that Eve can recover this vector $\mathbf{x}'$ implies that
Eve can in fact learn a non-negligible amount of information about
Alice's message. 

In order to more precisely determine the length of the basis which
Eve can practically find, one can consider the current state-of-the-art
for lattice basis reduction algorithms, For a study on limits of lattice
basis reduction algorithms, see \cite{MiccW}, for a more general
overview see \cite{MR09}. In \cite{MiccW}, the authors suggest that
the closest Eve can decode relative to the original message (given
a realistic amount of computational resources) is bounded by $\|\mathbf{x'}-\hat{\mathbf{x}}\|>2^{2\sqrt{n\log M\log1.005}}$.
According to \cite{MiccW}, this is an extremely conservative estimate,
and we claim that it would require either an infeasible amount of
computational resources or a substantial advancement in lattice-basis
reduction algorithms for Eve to learn more about the original message
than this bound indicates.

In order to prevent Eve from being able to recover a coded version
of the message $\mathbf{x}$, we apply a code which has a minimum
distance $d_{min}\leq2^{2\sqrt{n\log M\log1.005}}$. With high probability,
Eve will be unable to recover an error-free estimate of $\mathbf{x}$,
and will be unable learn any information about the message $\mathbf{m}$
sent using Algorithm 1.

\subsubsection{Correctness}

We now show that, if Alice encodes the output of Algorithm 1 with
an error-correcting code of $d_{min}=2^{2\sqrt{n\log M\log1.005}}$,
then given a sufficient number of recieve antennas, Bob will be able
to receive and decode the vector $\mathbf{x}$ without error, thus
allowing him to recover the message $\mathbf{m}$ from Algorithm 1.
We proceed by choosing the parameters $M=2^{n\log\log n/\log n},m=1,\alpha=1$,
and $k=1$, but this analysis is easily generalized to other parameters
which meet our security condition.

Assume that Bob has $n_{r}=t\cdot n$ receiever antennas for some
$t>2$. With probability at most $2e^{-n/2}$, the smallest singular
value in the channel between Alice and Bob will be at least $\sqrt{(t-2)n}$
(see for example \textbackslash{}cite\{rudelson\}, Eq. 2.3). This
implies that the noise variance in each of the decomposed channels
between Alice and Bob will decrease by at least of a factor of $\sqrt{(t-2)n}$
compared to the ambient channel noise. For Bob to be able to decode,
we now require the following condition to hold with some reasonable
probability:

\begin{equation}
\frac{M\alpha}{\sqrt{(t-2)n}}<2^{2\sqrt{n\log M\log1.005}},
\end{equation}

where the left-hand side is an upper bound of the noise variance in
each of Alice and Bob's parallel channels, and the right-hand side
is the minimum distance of the error correcting code employed by Alice.
Substituting the values of our constants, this gives us:

\begin{equation}
t>\frac{2.86}{n}
\end{equation}
and thus, for any $t>2$, the probabily of having a symbol error when
Bob decodes is bounded below by:

\begin{equation}
\mathrm{erfc}\left(\frac{n_{r}}{2.86}\right)=\mathtt{negl}(n),
\end{equation}
where $\mathrm{erfc}(\cdot)$ represents the complimentary error function.
We note that the condition that $t>2$ is only neccessary so that,
with overhelming probability, there will be no ill conditioned channels
in the parallel decomposition between Alice and Bob. For large $n$,
the probability of error will remain small for values of $t$ close
to 1.

\subsection{Secret Key Agreement and Computational Secrecy Capacity}

We now proceed to describe a secret key transmission protocol. In
the context of our model in Figure 1, such a protocol allows Alice
to transmit a ``secret'' to Bob in the form of bits undecodable by
Eve. Information-theoretic protocols which achieve this secret key
transmission based on equivocation at the eavesdropper are known,
for example \cite{SecretKeyAgreement}. Here we present a protocol
in which the key is kept secret under computational rather than information-theoretic
assumptions. We call the number of bits Alice can securely transmit
to Bob per channel use under computational assumptions the \emph{Computational
Secrecy Capacity}. 

Alice transmits a message $m\in[0,M)^{n}$ chosen uniformly at random
through her channel. This message has $n\log M$ bits of entropy.
By similar reasoning to the previous subsection, Eve receives this
message through her channel and recovers a vector, $m'$, such that
$\|m-m'\|>2^{2\sqrt{n\log M\log1.005}}$. The fact that Eve can recover
this vector $m'$ implies that Eve can in fact learn a non-negligible
amount of information about Alice's message. Here we argue that, in
the context of our proposed schemes, this information is in fact useless
unless Eve can violate our computational assumptions.

Achieving this bound in terms of Eve's ability to decode the transmitted
message implies that the entropy associated with Eve's estimate of
$m$ is at least $2\sqrt{n\log M\log1.005}$. By applying the leftover
hash lemma \cite{ILL89}, it follows that Alice can transmit to Bob
$2\sqrt{n\log M\log1.005}$ bits of information per message without
Eve learning any (non-negligible) amount of information about the
message. Thus, the security of the bits transmitted via the scheme
described in Algorithm 2 immediately holds. The quantity $2\sqrt{n\log M\log1.005}$
is the Computational Secrecy Capacity of our model in Figure 1. Table
II gives examples of the computational secrecy capacity for various
parameters.

Algorithm 2 requires the use of a universal hash function, which we
briefly define as a function which takes an arbitrary-sized input
and returns a fixed-size output (the hash value). A universal hash
is constructed so that there is a low expectation that two distinct
inputs chosen adversarially hash to the same value (collide). The
leftover hash lemma \cite{ILL89} essentially states that using a
universal hash function allows us to extract randomness from a source
from which an adversary has partial information, such that the adversary
is left almost no information about the output of the hash function.

\begin{algorithm}[h]
\caption{Key-Agreement Scheme}

Alice wishes to send Bob $\eta$ secret bits. Alice generates some
number $c=c(n)$, such that $2c\,\sqrt{n\log M\log1.005}>\eta$, of
random messages $m\in[0,M)^{n}$ and sends them to Bob over the MIMO
channel with channel parameters meeting the constraints in Theorem
1. Alice and Bob ensure that the message is exchanged without error
for example through channel coding. Alice and Bob then hash the message
(after decoding if channel coding is used), using a universal hash
which outputs $\eta$ bits and use the result as their secret. 
\end{algorithm}

We note that if Alice and Bob use channel coding to ensure for reliable
communications, then the rate of the coding scheme employed would
also effect the number of bits Alice and Bob could securely exchange
per channel use. That is, if Alice and Bob used a scheme that could
correct up to $e$ bit errors, then the bound on $\eta$ must be reduced
by $e$ bits. Similarly, this quantity, $2\sqrt{n\log M\log1.005}$,
also serves maximum number of bit errors that can be allowed to be
corrected in any code chosen between Alice and Bob in the scheme given
in Section V.B. 

\begin{table}
\caption{Computational Secrecy Capacity}

\begin{centering}
\begin{tabular}{|c|c|c|c|}
\hline 
$n$ & $\log_{2}M$ & SNR (dB) & Computational Secrecy Capacity\tabularnewline
\hline 
80  & 33.7 & 87.1 & 8.80\tabularnewline
\hline 
128  & 51.3 & 139.2  & 13.75\tabularnewline
\hline 
196  & 75.4 & 210.7  & 20.62\tabularnewline
\hline 
256 & 96 & 272.2 & 26.60\tabularnewline
\hline 
\end{tabular}
\par\end{centering}
\smallskip{}

For unit channel gain variance, the minimum constellation size, and
the minimum SNR that meets the noise requirements for the hardness
condition to hold. Computational Secrecy Capacity gives the number
of bits Alice can securely transmit to Bob per channel use using Algorithm
1.
\end{table}

\subsection{Computational Attacks}

We wish to briefly note the importance of maintaining the security
parameters as stated in Theorem 1. When the minimum noise requirement
is not met, it is possible that attacks follow on our system. By attack,
we mean that an adversary could, by applying a sub-exponential algorithm
in order to decode. As an example, in \cite{SS15}, the authors show
an attack on a version of our system with security parameters not
meeting those defined in Theorem 1. Finally, we note that it is not
apparent that an attack follows immediately for smaller values of
$M$ than required in Theorem 1, but we leave it as an open question
as to whether or not it can be shown that Eve can successfully decode
with non-exponential complexity for smaller values of M, or if a smaller
bound on the requirement for $M$ can be found that still entails
decoding to be hard for Eve.

\section{Conclusion}

We have demonstrated that the complexity of an eavesdropper decoding
a large-scale MIMO systems with M-PAM modulation can be related to
solving certain lattice problems which are widely conjectured to be
hard. This suggests that the complexity of solving these problems
grows exponentially with the number of transmitter antennas.  Unlike
the computationally hard problems underlying many of the most common
encryption methods used today, such as RSA and Diffie-Hellman, it
is believed that the underlying lattice problems are hard to solve
using a quantum computer, and thus this scheme presents a practical
solution to post-quantum cryptography.

It is not new to exploit properties of a communication channel to
achieve security; however, to our knowledge, this is the first scheme
which uses physical properties of the channel to achieve security
based on computational complexity arguments. Indeed, the notion of
the channel is not typically considered by cryptographers. We thus
describe our system as a way of achieving physical-layer cryptography.

Further novel to our scheme is the role that the channel gain matrix
plays in decoding. A transmitted message can only be decoded by a
user with the corresponding channel gain matrix. The channel gain
matrix, or more specifically the precoding of the message using the
right-singular vectors of the channel gain matrix, essentially plays
the role of a secret key in that it allows for efficient decoding
at the receiver. However, this value does not need to be kept secret,
nor does it play the traditional role of a public key. We term this
type of key as the Channel State Information- or CSI-key. In cryptography
terminology, this system is a trapdoor function, for which the trapdoor
varies both spatially and temporally. The fact that this is a new
type of cryptographic primitive suggests the possibility of entirely
new cryptographic constructions.

We have used the hardness result, in conjunction with a new notion
of computational secrecy capacity, to construct a method in which
two users can perform a key-agreement scheme, without a pre-shared
secret. In addition, we give a scheme that allows Alice and Bob to
securely communicate in the presence of an eavesdropper. We relate
the parameters required to maintain security to SNR requirements and
constellation size and show that they are practical to achieve assuming
a system with enough transmitter antennas and the corresponding number
of receivers, and relatively large constellation sizes.

\appendices{}

\setcounter{thm}{0}

\setcounter{conjecture}{0}

\appendices

\section{Proof of Main Theorem}
\begin{thm}
$\mathtt{MIMO-Search}_{M,\alpha,k}$ to $\mathtt{GapSVP}_{n/\alpha}$
and $\mathtt{SIVP}_{n/\alpha}$. Let $\alpha>0$, $m>0$ , $k>0$
be such that $m\alpha/k^{2}>\sqrt{n}$, and $M>m\,2^{n\log\log n/\log n}$.
Assume we have access to an oracle that solves $\mathtt{\mbox{\ensuremath{\mathtt{MIMO-Search}}}}_{M,\alpha,k}$,
given a polynomial number of samples from $A_{M,\alpha,k}$. Then
there exists an efficient quantum algorithm that given an $n$-dimensional
lattice $\mathcal{L}(\mathbf{A})$, solves the problems $\mathtt{GapSVP}_{n/\alpha}$
and $\mathtt{SIVP}_{n/\alpha}$. Additionally, there exists a classical
solution to $\mathtt{GapSVP}_{n/\alpha}$. 
\end{thm}
\begin{proof}
The lemmas required to prove this theorem are given below. We summarize
the proof of the main theorem as follows. 
\end{proof}
\begin{enumerate}
\item We first show that, given the MIMO oracle as described in Section
II, we can solve problems where the coefficients of the channel gain
matrix are instead drawn from a discrete Gaussian distribution as
described in Lemma 1. 
\item We begin with an arbitrary lattice basis $\mathbf{A}$ and apply the
LLL algorithm. We then create a discrete Gaussian distribution on
this lattice, with a second moment around the length of the largest
vector given in the reduced basis. We use this as the starting point
for the iterative portion of the algorithm. 
\item In Lemma 7, we use this MIMO decoding oracle to solve the $\mathtt{BDD}$
problem. The input to this problem is an (arbitrary) $n$-dimensional
lattice $\mathcal{L}(\mathbf{A})$, a number $r>\sqrt{2}\eta_{\epsilon}\left(\mathcal{L}(\mathbf{A})\right)$,
and a target point $\mathbf{y}$ within distance $d<M\sigma\alpha/k^{2}r\sqrt{2}$
(the bounding distance) of $\mathcal{L}(\mathbf{A})$. We take this
instance of a $\mathtt{BDD}$ problem and, using a number of samples
from the distribution $D_{L^{*},r}$, we are able to construct a number
of samples in the form of $\left(\mathbf{A},y=\left\langle \mathbf{a},\mathbf{x}\right\rangle +e\right)$,
in the exact form of distribution expected by MIMO decoding oracle
using Lemma 1. Here, returning the correct vector $\mathbf{x}$ solves
the $\mathtt{BDD}$ problem. We now have an oracle which solves the
$\mathtt{BDD}$ problem for arbitrary lattices. 
\item In Lemma 8, from \cite{R05}, the $\mathtt{BDD}$ oracle is used to
(quantumly) solve $\mathtt{DGS}_{L^{*},\sqrt{n}/\left(\sqrt{2}d\right)}$,
that is return samples of $D_{L^{*},r}$. Note we can efficiently
sample from $D_{L,r}$ for $r>\eta_{\epsilon}(L)$. If, in Lemma 7
we set parameters so that $\sqrt{2}d>\sqrt{n}$, then we can reduce
the value of $r$ to below the value for which we could previously
efficiently sample, that is we can construct a distribution that is
more narrow than previously possible. 
\item In \cite{R05}, the steps of Lemma 7 and 8 are iteratively applied,
resulting in a more narrow distribution of lattice points. Eventually,
this distribution becomes narrow enough to reveal information about
the shortest vectors of the lattice, solving the $\mathtt{GapSVP}_{n/\alpha}$
and $\mathtt{SIVP}_{n/\alpha}$. We refer the reader to \cite{R05}
for the rigorous treatment of this process. 
\item We refer the reader to \cite{P09} for the classical reduction, which
requires an oracle to solve the $\mathtt{BDD}$ problem. Replacing
Regev's $\mathtt{LWE}$-based $\mathtt{BDD}$ oracle with our MIMO-based
$\mathtt{BDD}$ oracle, the classical reduction follows. 
\end{enumerate}

\subsection{Smoothing Parameter.}

Before we prove the main theorem, we review the \emph{smoothing parameter}
and state some of its properties that we will require in our proof.
The smoothing parameter was introduced in \cite{MR04} and is an important
property of the behavior of a discrete Gaussian distribution on lattices.
It is precisely defined as follows. 
\begin{defn}
For an $n$-dimensional lattice $\mathcal{L}(\mathbf{\mathbf{A}})$
and a real $\epsilon>0$, the smoothing parameter, $\eta_{\epsilon}(\mathcal{L}(\mathbf{\mathbf{A}}))$,
is the smallest $\alpha$ such that $\Psi_{1/\alpha}\left(\mathcal{L}^{*}(\mathbf{\mathbf{A}})\backslash\left\{ 0\right\} \right)\leq\epsilon$.

The smoothing parameter defines the smallest standard deviation such
that, when the inverse is sampled over the dual $\mathcal{L}^{*}(\mathbf{A})$,
all but a negligible amount of weight is on the origin. More intuitively,
it is the width at which a discrete Gaussian measure begins to behave
as a continuous one. The motivation for the name `smoothing parameter'
is given in \cite{MR04}. For $\alpha>\sqrt{2}\eta_{\epsilon}(\mathcal{L})$,
if we sample lattice points from $D_{L,\alpha}$ then add Gaussian
noise $\Psi_{\alpha}$, then the resulting distribution is at most
distance $4\epsilon$ from Gaussian. We borrow the following two technical
claims from \cite{MR04}.
\end{defn}
\begin{claim}
\cite[Lemma 3.2]{MR04}. For any $n$-dimensional lattice $\mathcal{L}(\mathbf{A})$,
$\eta_{\epsilon}\left(\mathcal{L}\right)\leq\sqrt{n}/\lambda_{1}\left(\mathcal{L}^{*}(\mathbf{A})\right)$
where $\epsilon=2^{-n}$. 
\end{claim}
More generally, we can characterize the smoothing parameter for any
$\epsilon$. 
\begin{claim}
\cite[Lemma 3.3]{MR04}. For any $n$-dimensional lattice $\mathcal{L}(\mathbf{A})$
and $\epsilon>0$, 
\begin{equation}
\eta_{\epsilon}(\mathcal{L})\leq\sqrt{\frac{\ln(2n(1+1/\epsilon))}{\pi}}\cdot\lambda_{n}(\mathbf{A}).
\end{equation}
Equivalently, for any superlogarithmic function $\omega\left(\log n\right)$,
$\eta_{\epsilon}(\mathcal{L})\leq\sqrt{\omega\left(\log n\right)}\cdot\lambda_{n}(\mathbf{A}).$ 
\end{claim}
We will also note the following property of the smoothing parameter,
which follows from the linearity of lattices. If we scale the basis
of a lattice, then all of the successive minima will scale by the
same amount: 
\begin{claim}
For any $n$-dimensional lattice $\mathcal{L}(\mathbf{A})$,$\epsilon>0$,
and $c>0$,$\eta_{\epsilon}(\mathcal{L}\left(c\cdot\mathbf{A}\right))=c\cdot\eta_{\epsilon}(\mathcal{L}\left(\mathbf{A}\right))$. 
\end{claim}

\subsection{Preliminary Lemmas}

Before we can proceed with the main part of the proof showing the
reduction from standard lattice problems to the MIMO decoding problem,
we require several preliminary lemmas. We will first show, in Lemma
1, that it is sufficient to solve the MIMO decoding from a discrete
Gaussian distribution rather than a continuous one. This distribution
is used as input to the MIMO oracle in Lemma 7. We define the distribution
$D_{M,\alpha,k}$, which is the discrete analog of the distribution
$A_{M,\alpha,k}$. Given an arbitrary lattice $\mathcal{L}(\mathbf{A})$,
and a number $r>\sqrt{2}\eta_{\epsilon}(\mathcal{L}(\mathbf{A}))$,
we first sample a vector $\mathbf{a}$ from the distribution $D_{\mathcal{L}(\mathbf{A}),r}$,
and a point $e$ from the distribution $\psi_{\alpha}$. We now output:
\begin{equation}
\left(\frac{k\mathbf{a}}{r},\mathbf{y}=\left\langle \frac{k\mathbf{a}}{r},\mathbf{x}\right\rangle /M+e\right)
\end{equation}

\begin{lem}
Continuous-to-Discrete Samples. Given an oracle which can solve $\mathtt{\mbox{\ensuremath{\mathtt{MIMO-Decision}}}}_{M,\alpha,k}$,
there exists an efficient algorithm to recover $\mathbf{x}$ given
samples from $D_{M,\alpha,k}$.
\end{lem}
\begin{proof} We first claim that every point in the distribution
$D_{\mathcal{L}(\mathbf{A}),r}$ is proportional to $\psi_{r}$, to
within a negligible amount. This effectively follows from the fact
that we choose $r$ to be larger than the square root of two times
the smoothing parameter of the lattice (formally, see equation 11
of \cite[Claim 3.9]{R05}).

We now create a unitary transformation that can be applied to both
$\mathbf{a}$ and $\mathbf{y}$, creating $\mathbf{\widetilde{a}}$
and $\mathbf{\widetilde{y}}$, so that the support of $\mathbf{\widetilde{a}}$
is effectively $\mathbb{R}^{n}$. W do not have to actually cover
all of $\mathbb{R}^{n}$, nor do we in fact have to come close to
doing so. Our algorithms, by assumption, can only approximate points
in $\mathbb{R}^{n}$ to within a factor of $2^{-n^{c}}$, for some
$c>0$. Thus, we only need to have a support in which no point is
more than a distance of $2^{-n^{c}}$ away from any point in $\mathbb{R}^{n}$.
Then, by the fact that we choose a unitary transformation, all norms
and inner products will be preserved, and we will achieve our desired
result. We describe an appropriate transformation as follows.

Take two samples from the discrete distribution, call them $(\mathbf{a}_{i},\mathbf{y}_{i})$
and $(\mathbf{a}_{j},\mathbf{y}_{j})$. Now generate two numbers $c_{i},c_{j}\in\mathbb{Z}_{2^{n^{c}}}$,
uniformly, and output 
\[
\left(\frac{c_{i}\mathbf{a}_{i}+c_{j}\mathbf{a}_{j}}{c_{i}+c_{j}},\frac{c_{i}\mathbf{y}_{i}+c_{j}\mathbf{y}_{j}}{c_{i}+c_{j}}\right)=
\]
\[
\left(\frac{c_{i}\mathbf{a}_{i}+c_{j}\mathbf{a}_{j}}{c_{i}+c_{j}},\left\langle \frac{c_{i}\mathbf{a}_{i}+c_{j}\mathbf{a}_{j}}{c_{i}+c_{j}},\mathbf{x}\right\rangle +\frac{c_{i}e_{i}+c_{j}e_{j}}{c_{i}+c_{j}}\right)
\]
since each $e$ is generated i.i.d., it is not hard to see that the
noise has the correct distribution. Similarly, since each $\mathbf{a}$
is i.i.d., moments of the quantity $\frac{c_{i}\mathbf{a}_{i}+c_{j}\mathbf{a}_{j}}{c_{i}+c_{j}}$
will be unchanged, and this quantity will be proportional to the desired
Gaussian distribution. We have now increased the support of the distribution.
Indeed, the support of the quantities $\frac{c_{i}}{c_{i}+c_{j}}$
and $\frac{c_{j}}{c_{i}+c_{j}}$ covers every point in the interval
$[0,1)$ to within a factor of $2^{-n^{c}}$, and this new distribution
will be indistinguishable from the distribution expected by the MIMO
oracle. \end{proof} 

We next state the following claim which is proven in \cite{R05} and
shows that a small change in $\alpha$ results in a small change in
the distribution of $\Psi_{\alpha}$. This claim is required in the
proofs of Lemmas 2 and 3. 
\begin{claim}
\cite[Claim 2.2]{R05}. For any $0<\beta<\alpha\leq2\beta$, 
\begin{equation}
\Delta(\Psi_{\alpha},\Psi_{\beta})\leq9\left(\frac{\alpha}{\beta}-1\right).
\end{equation}

We next show that given a vector $\mathbf{x}$, it is easy to verify
whether or not it is the correct solution to the MIMO-Search problem:
\end{claim}
\begin{lem}
Verifying solutions of $\mathtt{MIMO-Search}_{M,\alpha,k}$ . There
exists an efficient algorithm that, given $\mathbf{x}'$ and a polynomial
number of samples from $A_{x,\alpha,k}$, for an unknown $\mathbf{x}$,
outputs whether $\mathbf{x=x'}$ with overwhelming probability.
\end{lem}
\begin{proof} Let $\xi$ be the distribution on $y-\mathbf{\mathbf{\left\langle a,x'\right\rangle }}$.
The same distribution can be obtained by sampling $e\sim\Psi_{\alpha}$
and outputting $e+\left\langle \mathbf{a,x-x'}\right\rangle $. In
the case $\mathbf{x=x'}$, this reduces to $e$, and the distribution
on $\xi$ is exactly $\Psi_{\alpha}$. In the case where $\mathbf{x\neq x'}$,
$\mathbf{\left\Vert x-x'\right\Vert }>1$ by the restriction on our
choices of $\mathbf{x}$. The inner product of $\left\langle \mathbf{a,x-x'}\right\rangle $
is Gaussian with zero mean and a standard deviation of at least $k/\sqrt{2\pi}$,
and the standard deviation on $e+\left\langle \mathbf{a,x-x'}\right\rangle $
must be at least $\sqrt{\alpha^{2}+k^{2}}/\sqrt{2\pi}$. We now must
distinguish between the random variables of $\Psi_{\alpha}$ and $\Psi_{\sqrt{\alpha^{2}+k^{2}}}$.

Assuming that $k^{2}$ is non-negligible in $n$, then given an arbitrary
number of samples within a polynomial factor of $n$, we can distinguish
between the two distributions with overwhelming probability by estimating
the sample standard deviation. \end{proof}

The following lemma is used from \cite[Lem. 3.7]{R05}. In \cite{R05},
the lemma applies to the case of LWE over an integer field, this proof
is repeated in this appendix in order to demonstrate that it follows
for the case of MIMO channels, given Lemma 7. Specifically, this lemma
shows that if we can solve the MIMO problems with noise parameter
$\alpha$, then we can solve the problems given samples with noise
drawn according to $\Psi_{\beta}$for any $\beta\leq\alpha.$ 
\begin{lem}
\emph{\cite[Lem 3.7]{R05}. Error Handling for $\beta\leq\alpha$.
Assume we have access to an oracle which solves $\mathtt{MIMO-Search}_{M,\alpha,k}$
by using a polynomial number of samples. Then there exists an efficient
algorithm that given samples from $A_{x,\beta,k}$ for some (unknown)
$\beta\leq\alpha$, outputs $\mathbf{x}$ with overwhelming probability.}
\end{lem}
\begin{proof} Assume we have at most $n^{c}$ samples for some $c>0$.
Let $Z$ be the set of all integer multiples of $n^{-2c}\alpha^{2}$
between 0 and $\alpha^{2}$. For each $\gamma\in Z$, do the following
$n$ times. For each sample, add a small amount of noise sampled from
$\Psi_{\sqrt{\gamma}}$, which creates samples in the form $A_{x,\sqrt{\beta^{2}+\gamma},k}$.
Apply the oracle and recover a candidate $x'$. Use Lemma 2 and check
whether $x'=x$. If yes, output $x'$, otherwise continue.

We now show the correctness of this algorithm. By Lemma 2, a result
can be verified to be a correct solution of $\mathtt{\mbox{\ensuremath{\mathtt{MIMO-Search}}}}_{M,\alpha,k}$
with probability exponentially close to 1. Thus we must only show
that in one iteration of the algorithm, we output samples that are
close to $A_{x,\alpha,k}$. Consider the smallest $\gamma\in Z$ such
that $\gamma\geq\alpha^{2}-\beta^{2}$. Then $\gamma\leq\alpha^{2}-\beta^{2}+n^{-2c}\alpha^{2}$.
And 
\[
\alpha\leq\sqrt{\beta^{2}+\gamma}\leq\sqrt{\alpha^{2}+n^{-2c}\alpha^{2}}\leq\left(1+n^{-2c}\right)\alpha.
\]

And by Claim 5, $\Delta\left(\Psi_{\alpha},\Psi_{\sqrt{\beta^{2}+\gamma}}\right)\leq9n^{-2c}$,
which is negligible in $n$. \end{proof} Standard lattice problems
are formulated with coefficient vectors that span all integers. In
our definition, we have limited our constellation size. Regev in \cite{R05}
introduces a variant on $\mathtt{BDD}_{\mathcal{L}(\mathbf{A}),d}$,
which we designate as $\mathtt{BDD}_{\mathcal{L}(\mathbf{A}),d}^{(M)}$.
This problem is identical to the $\mathtt{BDD}_{\mathcal{L}(\mathbf{A}),d}$
problem, with the exception that the coefficient vectors of the solution
are reduced modulo $M$, for arbitrary $M$. Regev shows that if we
can solve this variant of the problem in polynomial time, there in
fact exists a polynomial time algorithm which solves $\mathtt{BDD}_{\mathcal{L}(\mathbf{A}),d}$
in the general case. Thus for further lemmas, we can ignore the effect
of the limited constellation size. 
\begin{lem}
\cite[Lem. 3.5]{R05}. Finding coefficients modulo $M$ is sufficient.
Given a lattice $\mathcal{L}(\mathbf{A})$, a number $d<\lambda_{1}(\mathcal{L}(\mathbf{A})/2$,
and an integer $M\geq2$, access to an oracle which solves $\mathtt{BDD}_{\mathcal{L}(\mathbf{A}),d}^{(M)}$,
there exists an efficient algorithm that solves $\mathtt{BDD}_{\mathcal{L}(\mathbf{A}),d}$. 
\end{lem}
The following lemma shows that when sufficient noise is added to a
discrete Gaussian variable, it behaves like a continuous one. This
establishes a formal notion of the structure of the lattice being
`statistically hidden' by the noise. This lemma is used to show that
the distribution constructed in the proof of the main theorem is negligibly
close to the distribution required by the $\mathtt{MIMO-Search}$
oracle. 
\begin{lem}
\cite[Cor. 3.10]{R05}. For a lattice $\mathcal{L}(\mathbf{A})$,
vectors $\mathbf{z,u}\in\mathbb{R}^{n}$ , and two reals $r,\alpha>0$.
Assume that $1/\sqrt{1/r^{2}+\left(P/\alpha\right)^{2}}\geq\eta_{\epsilon}(\mathbf{A})$
for some $\epsilon<\frac{1}{2}$. Then the distribution of $\left\langle \mathbf{z,v}\right\rangle +e$,
where $\mathbf{v}$ is distributed according to $D_{L+\mathbf{u},r}$,
the norm of $\mathbf{z}$ is constrained to $P$, and $e$ is a normal
variable with zero mean and standard deviation $\alpha/\sqrt{2\pi}$,
is within total variational distance $2\epsilon$ of a normal variable
with zero mean and standard deviation $\sqrt{(rP)^{2}+\alpha^{2}}/\sqrt{2\pi}$. 
\end{lem}
We now state the following claim about a polynomial-time lattice basis
reduction algorithm, the LLL algorithm, given in $\cite{LLL82}$ and
improved by Schnorr in $\cite{S87}$. 
\begin{claim}
\cite{LLL82,S87}. For some $\mathcal{L}\left(\mathbf{A}\right)$,
we apply Schnorr's variant of the LLL algorithm, and obtain a new,
shorter basis for this lattice $\mathbf{\widetilde{A}}$. The norms
of the new basis vectors in this lattice, given by $\sigma_{1},...,\sigma_{n}$,
are bounded by: 
\begin{equation}
\sigma_{n}<2^{n\log\log n/\log n}\lambda_{n}
\end{equation}

and 
\begin{equation}
\sigma_{1}<2^{n\log\log n/\log n}\lambda_{1}
\end{equation}
\end{claim}
The rest of this proof proceeds with the following assumption:

\begin{equation}
1\leq\frac{\lambda_{n}}{\lambda_{1}}<2^{n\log\log n/\log n}
\end{equation}

The lower bound is evident from Minkowski's bound, and the upper bound
comes from the fact that we have reduced the basis by applying the
LLL algorithm. While no such upper bound would exist on an arbitrary
lattice, were this ratio to be bigger than this bound, then the LLL
algorithm would have returned exactly the shortest basis and we would
have already exactly solved the $\mathtt{GapSVP}$ and $\mathtt{SIVP}$
problems.

\subsection{Reducing MIMO Decoding to Standard Lattice Problems}

We now begin the main procedure of the reduction. We begin by taking
the basis and applying the LLL algorithm. 
\begin{lem}
We start with an arbitrary lattice $\mathcal{L}(\mathbf{A})$, and
apply the LLL algorithm to the basis $\mathbf{A}$. We now the use
procedure given in \cite{GPV08} and Schnorr's variant of LLL to get
a distribution $D_{L,r}$, for some $r>2^{n\log\log n/\log n}\lambda_{n}$ 
\end{lem}
The following lemma is the main mathematical contribution of this
work, and allows the MIMO oracle to be used in place of the LWE oracle
in the framework of Regev's reduction for LWE. From Figure 3, this
replacement implies that an efficient solution of the MIMO decoding
problem would also provide an efficient solution for standard lattice
problems. We briefly restate notation defined in Section II.C: $A_{M,\alpha,k}$
is the distribution of channel gains and the received signal from
a single antenna in a MIMO system, and $D_{A,\alpha}$ is the discrete
Gaussian distribution drawn over lattice $A$ with variance proportional
to $\alpha$.
\begin{lem}
\emph{$\mathtt{MIMO-Search}_{M,\alpha,k}$ to $\mathtt{BDD}_{L,r}$.
Let $\alpha>0$, $k>0$, $m>0$, and $M>m\,2^{n\log\log n/\log n}$.
Assume we have access to an oracle that, for all $\beta\leq\alpha$,
finds $\mathbf{x}$ given a polynomial number of samples from $A_{M,\beta,k}$
(without knowing $\beta$). Then there exists an efficient algorithm
that given an $n$-dimensional lattice $\mathcal{L}(\mathbf{A})$,
a number $r>\sqrt{2}\eta(\mathcal{L}(\mathbf{A}))$, and a target
point $\mathbf{y}$ within distance $d<M\sigma\alpha/\left(k^{2}r\sqrt{2}\right)$
of $\mathcal{L}(\mathbf{A})$, where $\sigma$ is the smallest eigenvalue
of $\mathbf{A}^{T}\mathbf{A}$, returns the unique $\mathbf{x}\in\mathcal{L}(\mathbf{A})$
closest to $\mathbf{y}$ with overwhelming probability.}
\end{lem}
\begin{proof} We describe a procedure that, given $\mathbf{y}$,
outputs a polynomial number of samples from the distribution of $D_{x,\beta,k}$.
Then, using the MIMO-Search oracle returns the closest point $\mathbf{x=As}$.

First, using Claim 1, sample a vector $\mathbf{v}\in\mathcal{L}^{*}\left(\mathbf{A}\right)$
from $D_{L^{*},r}$. We now output 
\begin{equation}
\left(k\mathbf{v}/r,k\left\langle \mathbf{v},\mathbf{A}^{-1}\mathbf{y}\right\rangle /rM+ke/r\right)
\end{equation}

We note that the right-hand side is equal to 
\[
k\left\langle \mathbf{v},\mathbf{A}^{-1}\mathbf{y}\right\rangle /rM+ke/r=
\]
\[
\left\langle k\mathbf{v}/r,\mathbf{s}\right\rangle /M+k\left\langle \mathbf{v},\mathbf{A}^{-1}\mathbf{\delta}\right\rangle /rM+ke/r
\]

We are guaranteed that $\left\Vert \delta\right\Vert <r$. We first
note that the quantity $k\left\langle \mathbf{v},\mathbf{A}^{-1}\mathbf{\delta}\right\rangle /r$
is distributed according to $D_{kL^{*}/r,\xi}$, where $\xi<kd/\sigma$.
We note that $\sigma$ relates to the maximum `skew' that results
from inverting the matrix $\mathbf{A}$. Since $\mathbf{\mathbf{A}}^{-1}\mathbf{\delta}$
is fixed for all samples, this `skew' is also fixed, meaning the inner
product $\left\langle \mathbf{v},\mathbf{A}^{-1}\mathbf{\delta}\right\rangle $
is symmetric since the distribution on $\mathbf{v}$ is symmetric. 

We now add some noise $e$ from the distribution $\psi_{\alpha}$
so that the discrete nature of $D_{kL^{*}/r,\xi}$ is effectively
washed out and we are left with a distribution that is essentially
Gaussian as expected by the MIMO oracle. That is the distribution
of the noise is within negligible total variational distance of $\psi_{\beta}$
for $\beta=\sqrt{\xi^{2}+\alpha^{2}/2}<\alpha$. This condition will
be true precisely when the condition given in Lemma 5 is met. We can
see that this holds: 
\begin{equation}
1/\sqrt{(1/k^{2}+(\sqrt{2}\xi/M\alpha)^{2})}\ge k/\sqrt{2}>\eta_{\epsilon}(\mathcal{L}^{*}(k\mathbf{A}/r))
\end{equation}

Therefore, the distribution in equation 16 is the distribution expected
by the MIMO oracle and we can recover the vector $\mathbf{x}$. \end{proof}
In order to relate the MIMO problem to standard lattice problems,
we need the following lemma, given in \cite{R05}, which uses a quantum
computer. 
\begin{lem}
\cite[Lem. 3.14]{R05}. $\mathtt{BDD}_{L,\alpha,k}$ to $D_{L^{*},\sqrt{n}/\left(\sqrt{2}d\right)}$.
There exists an efficient quantum algorithm that, given any $n$-dimensional
lattice L, a number $d<\lambda_{1}(\mathbf{A})/2$, and an oracle
that solves $\mathtt{BDD}_{L,d}$, outputs a sample from $D_{L^{*},\sqrt{n}/\left(\sqrt{2}d\right)}$. 
\end{lem}
In order for the reduction to hold, we must have $d>\sqrt{n}/\sqrt{2}$,
or the Gaussian distribution will actually grow in each iteration
of the procedure. We show this in the following claim. 
\begin{claim}
The inequality $M\sigma\alpha/\left(k^{2}r\sqrt{2}\right)>\sqrt{n}/\sqrt{2}$
is true given the constraints stated in Theorem 1.
\end{claim}
\begin{proof} We first set the constraint that $m\alpha/k^{2}>\sqrt{n}$,
and thus we only require that $M\sigma/mr>1$. We can see this is
true: 
\begin{equation}
\frac{M\sigma}{mr}>\frac{2^{n\log\log n/\log n}\sigma}{r}>\frac{2^{n\log\log n/\log n}\lambda_{1}}{\lambda_{n}}>1
\end{equation}
Here, the first step simply applies the bound given on $M/m$ in Theorem
1, and the second step follows from applying the bounds on $r$ and
$\sigma$ from Lemma 7 and Claim 6. The final step follows from applying
the bound in equation 13 on the quantity $\lambda_{n}/\lambda_{1}$.
\end{proof} Finally, we require the following lemma from \cite[Sec. 3.3]{R05},
which uses both the $\mathtt{BDD}$ and $\mathtt{DGS}$ oracles iteratively
to solve standard lattice problems. By setting $d=M\sigma\alpha/\left(k^{2}r\sqrt{2}\right)$,
we can iterate between the BDD and the DGS oracles, shrinking the
Gaussian distribution with each step, to a limit, and the stated standard
lattice problems. 
\begin{lem}
$\mathtt{DGS}_{L,\sqrt{n}/\left(\sqrt{2}d\right)}$ to standard lattice
problems. For $n$-dimensional lattice $\mathcal{L}(\mathbf{A})$,
$\alpha>0$,\textup{ $m>0$}, $k\in\mathbb{R}$ such that $m\alpha/k^{2}>\sqrt{n}$,
and $M>m\,2^{n\log\log n/\log n}$. Given an oracle which solves $\mathtt{BDD}_{L^{*},d}$
and $\mathtt{DGS}_{L,\sqrt{n}/\left(\sqrt{2}d\right)}$, then there
exists an efficient algorithm that given an $n$-dimensional lattice
$\mathcal{L}(\mathbf{A})$, solves the problems $\mathtt{GapSVP}_{n/\alpha}$
and $\mathtt{SIVP}_{n/\alpha}$. 
\end{lem}

\section{Complexity of Lattice Problems}

The security of any lattice-based cryptosystem is based on the presumed
hardness of lattice problems. In this subsection we limit our discussion
to the $\mathtt{GapSVP_{\gamma}}$ and $\mathtt{SIVP_{\gamma}}$ problems.

One well-known algorithm for solving lattice problems is the LLL algorithm
\cite{LLL82}. The algorithm solves $\mathtt{SIVP_{\gamma}}$ and
can be adapted to solving many other lattice problems as well. The
algorithm runs in polynomial time but only achieves an approximation
factor of $O\left(2^{n}\right)$. There have been a number of improvements
to this algorithm, such as Schnorr's algorithm \cite{S87} but none
that achieve small approximation factors (to within even a polynomial
factor of $n$) that run in polynomial time. All known algorithms
that return exact solutions to lattice problems in fact require a
running time on the order of $2^{n}$, see for example \cite{MV10}
or \cite{AKS01}. The hardness of these problems leads to the following
conjecture, stated in \cite[Conj. 1.1]{MR09}: 
\begin{conjecture}
There is no polynomial time algorithm that approximates lattice problems
to within an approximation factor that is within a polynomial factor
of $n$. 
\end{conjecture}
Within certain approximation factors, the complexity class of solving
lattice problems is known. It is NP-Hard to approximate $\mathtt{GapSVP}$
to within constant factors \cite{A96}, {[}8{]}\textendash {[}10{]}.
For a factor of $\sqrt{n}$, it belongs to the class NP$\cap$CoNP
\cite{GG00}, \cite{AR04}. It should be noted that our results are
based on an approximation factor of $n/c$. While such a strong hardness
result is not known for this regime, constructing algorithms to achieve
such approximation factors within polynomial time seems to be out
of reach. These results are summarized in Table III.

\begin{table}[th]
\centering{}\caption{Hardness results for standard lattice problems}

\begin{tabular}{|c|c|c|}
\hline 
$\gamma$  & Hardness  & Reference\tabularnewline
\hline 
\hline 
$O\left(1\right)$  & NP-Hard  & \cite{A96},\cite{K04},\cite{HR07},\cite{M01}\tabularnewline
\hline 
$\sqrt{n}<\gamma<n$  & NP$\cap$CoNP  & \cite{GG00},\cite{AR04}\tabularnewline
\hline 
$\gamma\sim n$  & \textendash{}  & \cite{R05},\cite{AD97}, This work\tabularnewline
\hline 
$2^{n\log\log n/\log n}$  & P  & \cite{LLL82},\cite{S87}\tabularnewline
\hline 
\end{tabular}
\end{table}

Another interesting result related to the hardness of lattice problems
considers quantum computation. If a quantum computer were to be realized,
this could have profound implications for the field of cryptography.
A quantum computer could efficiently factor numbers and solve the
discrete-logarithm problem, which would allow for virtually all key-exchange
protocols to be broken in polynomial time \cite{S97}. In addition,
algorithms such as Grover's search algorithm could improve exhaustive
searches by a factor of a square root, weakening the security of systems
like AES (the Advanced Encryption Standard, based on the Rijndael
cipher) \cite{G97}. There are currently no known quantum algorithms
that perform significantly better than the best known classical algorithms.
However, it should be noted that quantum algorithms are far less understood
and studied than classical algorithms, and thus we have less assurance
that such an algorithm does not exist. We will still state the following
conjecture, given in \cite[Conj. 1.2]{MR09}, but note that this is
a weaker conjecture than Conjecture 1: 
\begin{conjecture}
There is no polynomial time quantum algorithm that approximates lattice
problems to within polynomial factors. 
\end{conjecture}
Since we show in our main result that the hardness of decoding MIMO
is based on lattice problems, this implies that MIMO decoding cannot
be performed any faster on a quantum computer. In more general terms,
this means that lattice-based cryptography currently provides promise
for efficient constructions of classical cryptosystems that are secure
against quantum computers.

Besides being conjectured to be hard, many lattice problems have an
additional property that makes them attractive to cryptographers:
for certain problems, there are connections between the average-case
and worst-case complexities. This allows for the construction of systems
which are based on robust proofs of security. This property and its
significance is described next.

The worst-case complexity refers to the complexity of solving the
problem for the worst possible input of a fixed size; whereas average-case
complexity of a problem refers to the average complexity of solving
a problem given some underlying distribution of inputs of a fixed
size (typically uniformly random over all possible inputs). A worst-to-average
case reduction gives a distribution of inputs for which the average
complexity of solving a problem is as hard as the worst case complexity
(potentially of a different problem). The connections between worst-
and average-case complexity of certain lattice problems was first
found by Ajtai \cite{A96}. Ajtai constructed a function that is one-way
(that is, it can be computed in polynomial time, but is hard to invert)
on average based on the worst-case hardness of lattice problems. This
result was used by Ajtai and Dwork to construct a cryptosystem \cite{AD97}.
These worst-to-average case reductions were extended by many, but
most important to this paper is the reduction found in \cite{MR04}.

Basing cryptographic systems on problems where a worst-to-average
case reduction exists is an extremely strong guarantee of security.
It means that it is at least as hard to break the cryptosystem as
it is to solve \emph{any} instance of the related problems. Such a
strong guarantee is not provided by most cryptosystems today. For
example, breaking a cryptosystem that is based on factoring (e.g.
RSA) only implies a solution to factoring numbers of a specific form;
namely, the specific form used to generate RSA keys and not a solution
to worst-case factoring problems. 

Our scheme does not have an average-to-worst case reduction, and it
is not clear that such a reduction is possible. This means that there
may certain structured inputs which may make MIMO decoding easy \textendash{}
for example consider the case that $\mathbf{x}$ is sparse. Finding
a solution for decoding MIMO for sparse signals, would certainly not
lead to a solution to approximating all lattice problems to within
linear factors. A general polynomial-time algorithm that decodes all
MIMO signals under the conditions given in this paper, would, however
lead to a solution to approximating lattice problem, and this is unlikely. 

\section*{Acknowledgement}

The authors would like to thank Dan Boneh for his discussions on lattice-based
cryptography, Martin Hellman for his comments on a preliminary version
of this work, Shlomo Shamai for discussions regarding information-theoretic
secrecy in the context of our model, Mainak Chowdhury for discussions
on algorithms for MIMO decoding and linear codes, and Yonathan Morin
for his useful discussions and comments on a preliminary version of
this work.

\end{document}